\documentclass[11pt]{article}
\usepackage{fullpage}

\usepackage[dvipsnames]{xcolor}

\usepackage{mathptmx}

\usepackage{dsfont}
\usepackage{amsmath}
\usepackage{amsthm}
\usepackage{amssymb}
\usepackage{mathtools}

\usepackage{thmtools} 
\usepackage{thm-restate}

\usepackage{graphicx}
\usepackage{url}
\usepackage{nicefrac}
\usepackage{enumitem}
\usepackage{authblk}
\usepackage{tikz} 
\usepackage{subcaption}
\usepackage{xspace}
\usepackage[compact]{titlesec}
\usetikzlibrary{arrows,decorations.markings}
\usepackage[normalem]{ulem}
\usepackage{comment}

\usepackage{algorithm}
\usepackage[noend]{algpseudocode}

\usepackage{typearea}

\usepackage{pgf}
\usetikzlibrary{snakes}
\usetikzlibrary{patterns}
\usetikzlibrary{backgrounds}
\usetikzlibrary{scopes}
\usetikzlibrary{arrows, automata, positioning}
\usetikzlibrary{decorations.markings}
\usepackage[T1]{fontenc}
\usepackage{calc}
\usepackage{color}
\usepackage{lipsum}

\makeatother
\usepackage{color}

\definecolor{blueblack}{rgb}{0,0,.7}

\newcounter{sideremark}

\definecolor{Darkblue}{rgb}{0,0,0.4}
\definecolor{Brown}{cmyk}{0,0.61,1.,0.60}
\definecolor{Purple}{cmyk}{0.45,0.86,0,0}
\definecolor{brickred}{rgb}{0.8, 0.25, 0.33}
\usepackage[colorlinks,citecolor=blue,linkcolor=red,bookmarks=true]{hyperref}
\usepackage{cleveref}

\theoremstyle{plain}
\newtheorem{theorem}{Theorem}

\newtheorem{lemma}[theorem]{Lemma}

\theoremstyle{definition}

\theoremstyle{remark}

\newtheorem*{remark*}{Remark}

\newcommand{\NN}{{\mathbb N}}

\newcommand{\UFP}{\textsc{UFP}}
\newcommand{\bagUFP}{\textsc{bagUFP}}
\newcommand{\twUFP}{\textsc{twUFP}}
\newcommand{\spanUFP}{\textsc{spanUFP}}
\newcommand{\roundUFP}{\textsc{roundUFP}}

\newcommand{\eps}{\varepsilon}

\newcommand{\calA}{\mathcal{A}}

\newcommand{\calI}{\mathcal{I}}

\newcommand{\OPT}{\mathsf{OPT}}
\newcommand{\opt}{\mathsf{opt}}

\newcommand{\APX}{\mathsf{APX}}

\global\long\def\I{\mathcal{I}}%
\global\long\def\N{\mathbb{N}}%
\global\long\def\mi{\mathsf{mid}}%

\DeclareMathOperator{\poly}{poly}

\DeclareMathOperator{\tw}{tw}

\DeclareMathOperator{\pb}{pb}

\global\long\def\N{\mathbb{N}}%
\global\long\def\mi{\mathsf{mid}}%

\usepackage[colorinlistoftodos,textsize=tiny,textwidth=2cm,color=green!50!gray]{todonotes}

\ifdefined\DEBUG
  \newcommand{\aw}[1]{\textcolor{ForestGreen}{#1}}
  \newcommand{\ant}[1]{\textcolor{blue}{#1}}
  \newcommand{\fab}[1]{\textcolor{red}{#1}}
  \newcommand{\edi}[1]{\textcolor{orange}{#1}}
  \newcommand{\alex}[1]{\textcolor{violet}{#1}}

\newcommand{\alexnote}[1]{\todo[color=violet!100!black!50]{AL: #1}}
\newcommand{\anote}[1]{\todo[color=blue!100!black!50]{AT: #1}}
\newcommand{\awr}[1]{\todo[color=ForestGreen!100!black!50]{AW: #1}}
\newcommand{\fabr}[1]{\todo[color=red!100!black!50]{F: #1}}
\newcommand{\enote}[1]{\todo[color=orange!100!black!50]{E: #1}}

\else
  \newcommand{\ant}[1]{#1}
  \newcommand{\fab}[1]{#1}
  \newcommand{\edi}[1]{#1}
  
  \newcommand{\aw}[1]{#1}
  \newcommand{\alex}[1]{#1}

  \newcommand{\anote}[1]{}
  \newcommand{\fabr}[1]{}
  \newcommand{\enote}[1]{}
  \newcommand{\mnote}[1]{}
  \newcommand{\awr}[1]{}
  \newcommand{\alexnote}[1]{}
\fi

\author[1]{Alexander Armbruster}
\author[2]{Fabrizio Grandoni}
\author[2]{Edin Husi\'{c}}
\author[2]{Antoine Tinguely}
\author[1]{Andreas Wiese}

\affil[1]{Technical University of Munich, Munich, Germany

\texttt{alexander.armbruster@tum.de}, \texttt{andreas.wiese@tum.de}}
\affil[2]{USI-SUPSI, IDSIA, Lugano, Switzerland\thanks{Fabrizio Grandoni, Edin Husi\'{c} and Antoine Tinguely were supported by the Swiss National Science Foundation (SNSF) Grant 200021 200731/1.}

\texttt{fabrizio.grandoni@gmail.com}, \texttt{edinehusic@gmail.com}, \texttt{antoine.tinguely@idsia.ch}}

\title{On the Approximability of\\ Unsplittable Flow on a Path with Time Windows}

\date{}

\begin{document}

\maketitle   

\begin{abstract}
In the Time-Windows Unsplittable Flow on a Path problem~(\twUFP) we are given a resource whose available amount changes over
a given time interval (modeled as the edge-capacities of a given
path $G$) and a collection of tasks. Each task is
characterized by a demand (of the considered resource), a profit,
an integral processing time, and a time window. Our goal is to compute
a maximum profit subset of tasks and schedule them non-preemptively
within their respective time windows, such that the total demand of
the tasks using each edge $e$ is at most the capacity of $e$.  

We prove that \twUFP\ is $\mathsf{APX}$-hard
which contrasts the
setting of the problem without time windows, i.e., Unsplittable Flow
on a Path (\UFP), for which a PTAS was recently discovered {[}Grandoni,
M{ö}mke, Wiese, STOC~2022{]}. Then, we present a quasi-polynomial-time
$2+\eps$ approximation for \twUFP\ under resource augmentation.
Our approximation ratio improves to $1+\eps$ if all tasks' time windows
are identical. Our $\mathsf{APX}$-hardness holds also for this special
case and, hence, rules out such a PTAS \ant{(and even a QPTAS, unless $\mathsf{NP}\subseteq\mathrm{DTIME}(n^{\mathrm{poly}(\log n)})$)} \emph{without} resource augmentation.

\end{abstract}


\section{Introduction}

In the well-studied Unsplittable Flow on a Path problem (\UFP), we
are given a path graph $G=(V,E)$ with $m$ edges, a capacity $u(e)\in\N$
for each edge $e\in E$, and a collection $T$ of $n$ tasks. Each
task $i\in T$ is characterized by a \emph{weight} (or \emph{profit})
$w(i)\in\mathbb{N}$, a \emph{demand} $d(i)\in\mathbb{N}$, and a
subpath $P(i)\subseteq E$ of $G$.%
\footnote{Given a subpath $P$ of $G$, we will sometimes use $P$ also to denote
the corresponding set of edges $E(P)$. The meaning will be clear
from the context.}
A feasible solution consists of a subset of tasks $S\subseteq T$
such that $\sum_{i\in S:e\in P(i)}d(i)\leq u(e)$ for each $e\in E$.
In other words, the total demand of the tasks in $S$ whose subpath
contains $e$ does not exceed the capacity of $e$. Our goal is to
compute a feasible solution $\OPT$ of maximum profit $w(\OPT)\coloneqq\sum_{i\in\OPT}w(i)$.
One can naturally interpret $G$ as a time interval subdivided into
time slots (the edges). At each time slot, a given amount of a considered
resource (e.g., energy) is available. Each task $i$ corresponds to
a job that we can execute (or not) in a fixed time interval: if $i$
is executed, it consumes a fixed amount of the considered resource
during its entire execution and generates a profit of $w(i)$. \UFP\ is
strongly NP-hard \cite{BSW14,CWMX10}, and a lot of attention was
devoted to the design of approximation algorithms for it \cite{AGLW14,BCES2006,BFKS09,BGKMW15,BSW14,GMW21,GMW22,GMWZ17,GMWZ18},
culminating in a recent PTAS for the problem~\cite{GMW22STOC}. 

In \UFP, we have no flexibility for the time interval during which
each selected task $i$ is executed. In practice, it makes sense to
consider scenarios where $i$ has a given length (or processing time)
and a \emph{time window }during which it needs to be executed. In
this setting, for each selected task $i$ we need to specify a starting time for $i$ such that $i$ is processed completely within its time
window. This leads to the Time-Windows UFP problem (\twUFP). 
Here, we are given the same input as in \UFP, with the difference that for each task $i$, instead
of a subpath $P(i)$ we are given a \emph{length} (or \emph{processing
time}) $p(i)\in\{1,\ldots,m\}$, and a subpath $\tw(i)\subseteq E$
with at least $p(i)$ edges (the \emph{time window} of $i$).
A \emph{scheduling} of $i$ is a subpath $P(i)\subseteq\tw(i)$ containing
precisely $p(i)$ edges. 
A feasible solution for the given instance is a pair $(S,P(\cdot))$
such that for each $i\in S$ the path $P(i)$ is a feasible scheduling
of $i$ and $\sum_{i\in S:e\in P(i)}d(i)\leq u(e)$ for each
$e\in E$. Our goal is to maximize $w(S)$, like in \UFP. Observe
that \UFP\ is the special case of \twUFP\ where for each task
$i\in T$ the time window $\tw(i)$ contains exactly $p(i)$
edges; hence, \twUFP\ is strongly $\mathsf{NP}$-hard. The best
known approximation algorithm for it is a $O(\log n/\log\log n)$-approximation
\cite{GIU15}, improving on a prior result in \cite{CCGRS14} (both
results hold for a more general problem, \bagUFP, which is defined
later). However, no result in the literature excludes the existence
of a much better approximation ratio for \twUFP, including possibly
a PTAS. In this paper, we make progress on a better understanding
of the approximability of \twUFP.

\subsection{Our Results and Techniques}

Our first main result is that \twUFP{} does \emph{not }admit a PTAS,
which in particular implies that it is strictly harder than \UFP{}
(unless $\mathsf{P=NP}$). We show that this already holds when the
time window of each task spans all the edges of $G$, i.e.,
if $\tw(i)=E$ for each $i\in T$; we denote this special case of
the problem by Spanning UFP (\spanUFP).
Our result even holds in the cardinality setting, i.e., when all tasks have unit weight, and for polynomially bounded
input data. 
\begin{theorem}
\label{thr:apxhard}
\spanUFP{} (thus also \twUFP{}) is $\mathsf{APX}$-hard and does not admit a (polynomial-time) $\frac{2755}{2754}$-approximation algorithm (unless $\mathsf{P=NP}$), even in the cardinality case and if demands and the number of edges is polynomially bounded in $n$.
\end{theorem}

\ant{The proof of Theorem~\ref{thr:apxhard} follows from a reduction to Maximum 3-Dimensional Matching (3-DM) and the bound in \cite{CC06}.
In our reduction, we model 3-DM instances as instances of \spanUFP{} in which the capacity profile models bins (see Figure~\ref{fig:hardnessCapacityProfile}),  reminiscent of the 2-Dimensional Vector Bin Packing (2-VBP) problem
, and adapt the construction in \cite{W97} (which refutes the existence of an asymptotic PTAS for 2-VBP).
}

\begin{figure}[t]
\centering \begin{tikzpicture}
		\def\w{1.8}
		\def\eps{0.01}
		\draw (0.1,0) -- (5*\w,0);
		\draw[thick, color=red] (0.1,0)--(0.1,2);
		\draw[thick, color=red] (0.1,2)--(\w/2+0.05,2);
		\draw[thick, color=red] (\w/2+0.05,2)--(\w/2+0.05,2-4*\eps);
		\draw[thick, color=red] (\w,2-4*\eps)--(\w/2+0.05,2-4*\eps);
		\draw[thick, color=red] (\w,2-4*\eps)--(\w,0);
		\foreach \n in {1,...,4}
		{\draw[thick, color=red] (\w*\n,0)--(\w*\n+0.1,0);
			\draw[thick, color=red] (\w*\n+0.1,0)--(\w*\n+0.1,2);
			\draw[thick, color=red] (\w*\n+0.1,2)--(\w*\n+\w/2+0.05,2);
			\draw[thick, color=red] (\w*\n+\w/2+0.05,2)--(\w*\n+\w/2+0.05,2-4*\eps);
			\draw[thick, color=red] (\w*\n+\w,2-4*\eps)--(\w*\n+\w*2/4+0.05,2-4*\eps);
			\draw[thick, color=red] (\w*\n+\w,2-4*\eps)--(\w*\n+\w,0);}
		\draw [->,>=latex] (-0.5,0) -- (-0.5,2.5);
		\draw (-0.6,0) -- (-0.4,0);
		\draw (-0.9,0) node  {$0$};
		\draw (-0.5,2.75) node  {capacity};
		\foreach \n in {1,...,90}
		{\draw node[circle,fill,inner sep=0.8pt] at (0.1*\n,0) {};
		}
    	\draw[] (5*\w +0.75, 1) node {$\cdots$};
	\end{tikzpicture} \caption{The capacity profile (red) of an instance \aw{used in our hardness
result for} \spanUFP{}.}
\label{fig:hardnessCapacityProfile}
\end{figure}
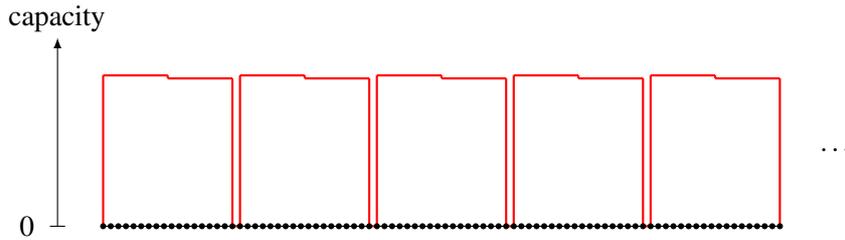

Our second main contribution is a constant factor approximation algorithm
for \twUFP\ with resource augmentation which runs in quasi-polynomial
time. More specifically, for any constant $\eps>0$, we compute a
$(2+\eps)$-approximate solution in time $2^{O_{\eps}(\poly(\log(nm)))}$
which violates the edge capacities at most by a factor $1+O(\eps)$.
As usual in the setting of resource augmentation, the approximation
ratio is computed with respect to\ an optimal solution which is \emph{not}
allowed to violate the edge capacities.

\begin{restatable}{theorem}{twapx} 
\label{thr:tw-apx} For any $\eps>0$, there is a 
$(2+\eps)$-approximation for \twUFP{} under $(1+O(\eps))$-resource
augmentation with running time $O\big( (mn)^{O(\log^2 n \log^5 m/\varepsilon^4)} \cdot \log U \big)$, where $U \coloneqq \max_{e\in E} u(e)$ is the maximum capacity of the instance.
\end{restatable} 

For the special case of \spanUFP{}, \ant{the approximation ratio of our algorithm improves to $1+\eps$}, i.e., we achieve
a \ant{quasi-PTAS (QPTAS)} under resource augmentation for \spanUFP{}. In contrast \ant{(unless $\mathsf{NP}\subseteq\mathrm{DTIME}(n^{\mathrm{poly}(\log n)})$), a QPTAS for \spanUFP{} is impossible \emph{without} resource augmentation since \spanUFP{} is $\mathsf{APX}$-hard.
}

\begin{restatable}{theorem}{spanqptas}
\label{thr:qptas} For any $\eps>0$ there is a QPTAS for \spanUFP{}
under $(1+O(\eps))$-resource augmentation. 
\end{restatable} 

We leave as an interesting open problem to find a polynomial-time
(or even quasi-polynomial-time) constant factor approximation for
\twUFP\ \emph{without} resource augmentation.

We describe now the key ideas in our $(2+\epsilon)$-approximation
for \twUFP{}. The previous QPTASes for \UFP{} \cite{BCES2006,BGKMW15}
are based
on the following approach (written in slightly different terminology
here). We interpret the path $E$ as an interval $I$, i.e., with
one subinterval of unit length for each edge $e\in E$. We consider
its midpoint which we denote by $\mi(I)$ and all input tasks $i$
whose path $P(i)$ contains $\mi(I)$. These tasks are partitioned
into polylogarithmically many groups such that within each group,
all tasks have roughly the same weight and demand (up to a factor
of $1+\epsilon$). For each group, we guess an estimate for the amount
of capacity they use in the optimal solution on each edge in $E$.
Given this, we compute the most profitable set of tasks for which
these edge capacities are sufficient. Then, the remaining problem
splits into two \emph{independent }subproblems: one for all input
tasks whose paths lie on the left of $\mi(I)$ and one for all input
tasks whose paths lie on the right of $\mi(I)$. Therefore, we can
easily recurse on these two subproblems and solve them independently. 

Unfortunately, this approach does not extend to \twUFP{}. A natural
adaption for \twUFP{} would be to consider all input tasks $i$
for which the path $P(i)$ \emph{in the optimal solution }contains
the midpoint of $I$ which we denote by $\mi(I)$. However, then the
remaining problem does \emph{not }split nicely into two independent
subproblems: there can be a task $i$ whose time window $\tw(i)$
contains $\mi(I)$ and which we could schedule entirely on the left
of $\mi(I)$ or entirely on the right of $\mi(I)$ (see Figure~\ref{fig:intro}).
Thus, $i$ appears in the input of the left \emph{and} the right subproblem,
even though we are allowed to select $i$ only once. Hence, these
subproblems are no longer independent. Even more, this issue for $i$
can arise again in each level of the recursion, which yields many
interconnected subproblems.

Therefore, we use a different approach to obtain a $(2+\epsilon)$-approximation
for \twUFP{}. We consider all input tasks $i$ whose \emph{time
window} contains $\mi(I)$. By losing a factor of 2 on the weight
of these tasks in the optimal solution, we can sacrifice either all
of them that are scheduled on the left of $\mi(I)$ or all of them
that are scheduled on the right of $\mi(I)$. We guess which case
applies and we assume it in the following to be the latter case w.l.o.g.
For the tasks $i$ for which $P(i)$ contains $\mi(I)$ in the optimal
solution we guess an estimate for the amount of capacity they use
on the edges, similar to the QPTASs for \UFP{}. Then we recurse
on the subproblems corresponding to the left and the right of $\mi(I)$.
For the right subproblem, we do not allow to select tasks whose time
window uses $\mi(I)$. Therefore, our two subproblems are now independent!
Let us consider the left subproblem, denote its corresponding interval
by $I^{L}$ and its midpoint by $\mi(I^{L})$. If we continued natively
in the same fashion, we would lose another factor of 2 on the weight
of the tasks whose time windows contain $\mi(I)$ and $\mi(I^{L})$,
which we cannot afford. Instead, we employ a crucial new idea. We
partition the mentioned tasks into polylogarithmically many groups
such that within each group, all tasks have the same demand (using
resource augmentation) and roughly the same weight (up to a factor
of $1+\epsilon$). Since the time window of each such task $i$ contains
$\mi(I)$ and $\mi(I^{L})$, we can schedule it freely within the
interval $[\mi(I^{L}),\mi(I)]$. Therefore, for each group, we consider
its tasks that are scheduled entirely during $[\mi(I^{L}),\mi(I)]$
in the optimal solution and round their processing times via linear grouping \cite{gonzalez2007handbook,vl81,williamson2011design} to $O(1/\epsilon)$ different values. For each resulting
combination of demand, weight, and rounded processing time, we guess
for the right subproblem of $I^{L}$ (i.e., the ``left-right subproblem'')
how many tasks it schedules and give it these tasks as part of the
input. To the left subproblem of $I^{L}$ (i.e., the ``left-left
subproblem'') we give the additional constraint that it needs to
leave the corresponding number of tasks from each group unassigned.
An important aspect is that the processing times of these tasks are
\emph{not} rounded in the left-left subproblem, but they \emph{are}
rounded in the left-right subproblem. This is a crucial difference
of our approach in comparison to other rounding methods from the literature.
Thanks to this technique, we avoid to lose another factor of 2 on
the mentioned tasks. We apply it recursively for $O(\log m)$ levels
and obtain an approximation ratio of $2+\epsilon$ overall.

\begin{figure}[t]
\centering
\begin{tikzpicture}
	\draw[|-|] (0,0) -- (10,0) ;
	
	\draw[-] (2.5,0.1) -- (2.5,-0.1) node[below] {$\mi(I^L)$};
	\draw[-] (5,0.1) -- (5,-0.1) node[below] {$\mi(I)$};
	
	\draw[|-|, dotted] (1,0.6) -- (8,0.6) node[right] {$tw(i)$};
	
	\node[above] at (4.2,0.6) {$P(i)$};
	
	\filldraw[fill=gray, draw=black] (3.7,0.5) rectangle (4.7,0.7);
	
	\draw[thick] (0,0.15) 
	.. controls (0.2, 0.7) and (2.5, 0) .. (2.5,0.7)
	.. controls (2.5, 0) and (4.8, 0.7) .. (5, 0.15);
	\node[above] at (2.5,0.6) {$I^L$};
\end{tikzpicture}
\caption{The interval $I$ and its subdivision together with a task $i$.}
\label{fig:intro}
\end{figure}
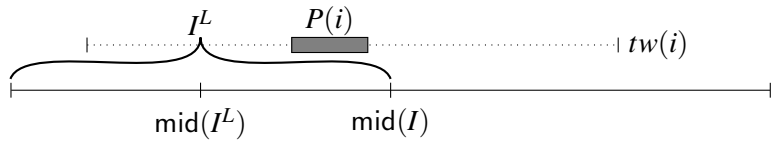

\subsection{Related work}

The \bagUFP\ problem is a generalization of \UFP\ where we are given the same input as in \UFP\ plus a partition of the tasks into subsets, the \emph{bags}, and we are allowed to select at most one task per bag. This also generalizes \twUFP\ by creating bags that model each possible way to schedule a task within its time window. It was already known that \bagUFP\ is APX-hard \cite{S99} (which is now also implied by our result). The current best approximation ratio for \bagUFP\ is $O(\log n/\log\log n)$ \cite{GIU15}, and $17$ under the no-bottleneck assumption \cite{CCGRS14}\ant{, i.e., if $\max_{i\in T} d(i) \leq \min_{e\in E} u(e)$}.
A constant factor approximation algorithm for \bagUFP\ (hence for \twUFP) is known for the cardinality case of the problem \cite{GIU15}.

For the special case of \twUFP\ when all the demands and capacities are $1$ \ant{(also known as the non-preemptive throughput maximization problem)},
the best-known approximation factor for this problem is~$2$~\cite{BGNS01,bar2001unified,BD00}. Notice that we give a $(2+\eps)$-approximation for a much more general problem, however, in quasi-polynomial time and using resource augmentation.
In the cardinality case of the problem, there is an algorithm with an approximation ratio of  $e/(e-1)+\eps$ \cite{COR06} (also for bags instead of time windows),
which was later slightly improved in \cite{ILM20}.
\anote{removed result with multiple machines}

\section{A $(2+\eps)$-approximation for \twUFP{}}\label{sec:algorithm}

In this section, we present a $(2+\eps)$-approximation algorithm
with a quasi-polynomial running time for \twUFP{} under $(1+\eps)$-resource
augmentation, for any constant $\eps>0$.

Without loss of generality, assume $1/\eps \in \NN$.
Formally, let $(\OPT,P^*(\cdot))$ denote an optimal solution for the given instance. We compute
a solution $(S,P(\cdot))$ consisting of a set of tasks $S$ and for each task $i \in S$ a subpath $P(i)\subseteq \tw(i)$ of length $p(i)$ such that $w(\OPT)\le(2+\eps)w(S)$
and $\sum_{i\in S:e\in P(i)} d(i) \le(1+\eps)u(e)$ for every $e\in E$.

First, we use resource augmentation in order to round the edge capacities
and task demands. We also round the tasks' weights.
\ant{This yields the following reduction, whose proof is presented in Appendix~\ref{sec:proof-preprocessing}.}

\begin{restatable}{lemma}{reductions}
\label{lem:preprocessing}Let $\alpha\ge 1$. Assume that there is
an $\alpha$-approximation algorithm running in time $T(n, \eps)$ for the special case of \twUFP{}
in which
\begin{itemize}
\item for each $e\in E$ we have that $u(e)\in[1,(n/\eps)^{1/\eps}))$
\item for each task $i\in T$ we have that
\begin{itemize}
\item $d(i)\in[1,(n/\eps)^{1/\eps})$ and $d(i)$ is a power of $1+\eps$,
and
\item $w(i)\in[1,n/\eps]$ and $w(i)$ is a power of $1+\eps$.
\end{itemize}
\end{itemize}
Then there is a $(1+5\eps)\alpha$-approximation algorithm for
\twUFP{} under $(1+ 4\eps)$-resource augmentation with a running time of $O(\poly(n) \frac{\log U}{\log n} T(n,\eps))$, where $U=\max_{e\in E}u(e)$.
\anote{replaced $\log_2$ with $\log$ for consistency}
\end{restatable}

In the following, we will present a $(2+\eps)$-approximation
algorithm for the special case of \twUFP{} defined in Lemma~\ref{lem:preprocessing};
this yields a $(2+O(\eps))$-approximation for arbitrary instances
of \twUFP{} under $(1+O(\eps))$-resource augmentation.

We first define a hierarchical decomposition of $E$. Assume w.l.o.g.
that $m=|E|$ is a power of 2. We define a (laminar) family of subpaths $\I$
which intuitively form a tree. We will refer them also as \emph{intervals}.
There is one interval $I_{r}\in\I$ such that
$I_{r}=E$ which intuitively forms the root of the tree. For each
$I\in\I$, we will ensure that $|I|$ is a power of 2. Consider an
interval $I\in\I$ with $|I|\ge2$ and let ${\mi(I)}$ denote its middle
vertex. Given $I$, we define recursively that there is an interval
$I^{L}\in\I$ induced by the edges of $I$ to the left of ${\mi(I)}$,
and an interval $I^{R}\in\I$ 
induced by the edges of $I$ to the right of ${\mi(I)}$.
We say that $I^{L}$ and $I^{R}$ are the two \emph{children}
of $I$. We apply this definition recursively. As a result, for each
edge $e\in E$ there is an interval $\{e\}\in\I$; such
intervals form the leaves of our tree. We say that an interval $I\in\I$
is of \emph{level $\ell$ }for some $\ell\in\N_{0}$ if $|I|=2^{\log m-\ell}$,
e.g., the interval $I_{r}$ is of level 0. Note that we have $1+\log m$
levels (for which there is an interval in $\I$). 

We group the tasks into a polylogarithmic number of groups. For each
task $i\in T$ we define that it is of \emph{level $\ell(i)$ }if
$\tw(i)$ is contained in an interval $I\in\I$ of level $\ell(i)$
but not in an interval $I'\in\I$ of level \ant{$\ell(i)+1$}. 
We define $T_{\ell}\coloneqq \{i\in T:\ell(i)=\ell\}$ for each integer $\ell$.
Also,
for each combination of values $w,d$ such that $w$ represents a weight and~$d$ a demand,
we define the set $T_{w,d}=\{i\in T: w(i)=w\wedge d(i)=d\}$. 
Note that there are only a polylogarithmic number of
combinations for these pairs $w,d$ for which $T_{w,d} \ne \emptyset$ \ant{by Lemma~\ref{lem:preprocessing}}.

Recall that $(\OPT,P^*(\cdot))$ denotes an optimal solution. Consider a
task $i\in\OPT$ of some level~$\ell(i)$. Let $I$ be the (unique)
interval of level $\ell(i)$ that contains $\tw(i)$ and let $I^{L}$
and $I^{R}$ denote the children of $I$.
Observe that $P^*(i)$ must satisfy one of the following three conditions:
\begin{enumerate}
\item\label{case:optScheduledLeft} $P^{*}(i)$ is contained in $I^{L}$, or
\item\label{case:optScheduledMiddle} $\mi(I)$ is in the interior of $P^{*}(i)$, i.e., $\mi(I)$ is a
vertex of $P^{*}(i)$ but it is neither the leftmost nor the rightmost
vertex of $P^{*}(i)$, or
\item\label{case:optScheduledRight} $P^{*}(i)$ is contained in $I^{R}$.
\end{enumerate}
Let $\OPT^{(L)},\OPT^{(M)},$ and $\OPT^{(R)}$ denote the tasks in
$\OPT$ for which case \ref{case:optScheduledLeft}, \ref{case:optScheduledMiddle}, and \ref{case:optScheduledRight} applies, respectively.
In particular,
$w(\OPT)=w(\OPT^{(M)})+w(\OPT^{(L)})+w(\OPT^{(R)})$ and thus $w(\OPT^{(L)})\le\frac{1}{2}w(\OPT)$
or $w(\OPT^{(R)})\le\frac{1}{2}w(\OPT)$. The first step of our algorithm
is to guess which of these cases applies. Assume w.l.o.g.\ that $w(\OPT^{(R)})\le\frac{1}{2}w(\OPT)$.
In the following, we will essentially aim to select only tasks from
$\OPT^{(L)}\cup\OPT^{(M)}$. 
In this step we lose a
factor of (up to) 2 on the obtained profit. 

Formally, we say that a solution $(S,P(\cdot))$ is \emph{left constrained} if the following holds for each task $i\in S$:
For $I$ being the unique interval of level $\ell(i)$ containing $\tw(i)$, we require that $P(i)$ is not contained in the right child of $I$, i.e., in the subinterval of $I$ containing all edges of $I$ on the right of $\mi(I)$.
We will later define artificial tasks (not contained in $T$) for which this is not required.
Our algorithm will return a left constrained solution by construction and we will compare its profit with the left constrained solution~$\OPT^{(L)}\cup\OPT^{(M)}$.

Our algorithm is recursive. We start by describing the root subproblem in the next subsection;
it is simpler than a generic subproblem, however it is convenient to introduce certain notions that will be needed also in the general case. In the subsequent subsection, we will describe a generic subproblem.

\subsection{Root subproblem}

The root subproblem
corresponds to the interval $I_{r}=E$ (of level~0). 
Let $(\overline \OPT_r, P^*(\cdot))$ be an optimal left constrained solution. Let $\overline \OPT^{(M)}_r$ be all the tasks $i \in \overline \OPT_r$
such that $\mi(I_r)$ is in the interior of $P^*(i)$.
Consider a group $T_{w,d}$. We want to guess a set of \emph{boxes} that delimit space that we reserve for tasks in $T_{w,d}\cap\overline \OPT^{(M)}_r$.
Formally, a box $b$ is characterized by a path $\pb(b)\subseteq E$,
a height $h(b)>0$, a demand $d(b)>0$, and a weight $w(b)>0$. We will consider only boxes $b$ for which $w(b)$ is some task weight, $d(b)$ is some task demand, and $h(b)$ is some integer multiple of $d(b)$ upper bounded by~$n\cdot d(b)$. In the following, we will assign
certain sets of tasks to some boxes such that, intuitively, these
tasks are stacked on top of each other inside the box and each of
them has a weight of $w(b)$, a demand of $d(b)$, and their total demand is at most $h(b)$, i.e., they are at most $h(b)/d(b)$ many. Furthermore, it is possible to schedule each such task within the path of the box\ant{, as illustrated in Figure~\ref{fig:box}.}

More explicitly, 
for any set of tasks $S$, using the notation $d(S)\coloneqq \sum_{i\in S}d(i)$ and $w(S)\coloneqq \sum_{i\in S}w(i)$, we say that $S$
\emph{fits inside a box }$b$ if 
\begin{itemize}
\item $d(S)\le h(b)$,
\item $|\tw(i)\cap\pb(b)|\ge p(i)$ for each $i\in S$, and
\item $w(i)=w(b)$ and $d(i)=d(b)$ for each $i\in S$.
\end{itemize}

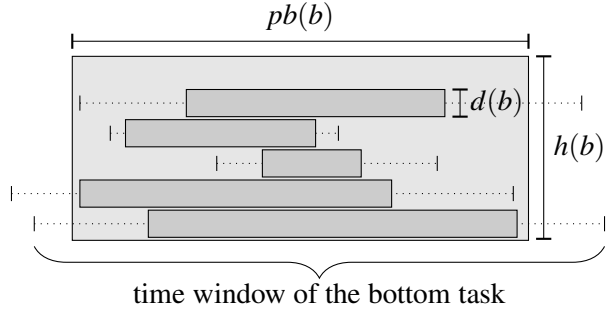
\begin{figure}[t]
	\centering
	\begin{tikzpicture}
	\draw [draw=black, fill=black!10!white] (-3, -0.02) rectangle (3,2.42);

	\draw[dotted] (-3.5,0.2) -- (4,0.2);
	\draw (-3.5,0.1) --  (-3.5,0.3);
	\draw (4,0.1) --  (4,0.3);
	\draw [black, fill=black!20!white] (-2,0.02) rectangle (2.85,0.38);

	\draw[dotted] (-3.8,0.6) -- (2.8,0.6);
	\draw (-3.8,0.5) --  (-3.8,0.7);
	\draw (2.8,0.5) --  (2.8,0.7);
	\draw [black, fill=black!20!white] (-2.9,0.42) rectangle (1.2,0.78);

	\draw[dotted] (-1.1,1) -- (1.8,1);
	\draw (-1.1,0.9) --  (-1.1,1.1);
	\draw (1.8,0.9) --  (1.8,1.1);
	\draw [black, fill=black!20!white] (-0.5,0.82) rectangle (0.8,1.18);

	\draw[dotted] (-2.5,1.4) -- (0.5,1.4);
	\draw (-2.5,1.3) --  (-2.5,1.5);
	\draw (0.5,1.3) --  (0.5,1.5);
	\draw [black, fill=black!20!white] (-2.3,1.22) rectangle (0.2,1.58);

	\draw[dotted] (-2.9,1.8) -- (3.7,1.8);
	\draw (-2.9,1.7) --  (-2.9,1.9);
	\draw (3.7,1.7) --  (3.7,1.9);
	\draw [black, fill=black!20!white] (-1.5,1.62) rectangle (1.9,1.98);

	\draw (3.2,0) -- node[right] {$h(b)$} (3.2,2.42);
	\draw[thick] (3.1,0)--(3.3,0) ;
	\draw[thick] (3.1,2.42)--(3.3,2.42) ;

	\draw (2.1,1.62) -- node[right] {$d(b)$} (2.1,1.98);
	\draw[thick] (2,1.62)--(2.2,1.62) ;
	\draw[thick] (2,1.98)--(2.2,1.98) ;

	\draw (-3,2.62) -- node[above] {$pb(b)$} (3,2.62);
	\draw[thick] (-3,2.52)--(-3,2.72) ;
	\draw[thick] (3,2.52)--(3,2.72) ;

	\draw[decorate,decoration={brace,amplitude=12pt}] (4,-0.1) -- (-3.5,-0.1) node[midway, below,yshift=-10pt,]{time window of the bottom task};

\end{tikzpicture}
\caption{A box of height $h(b)$, demand $d(b)$ and a path $pb(b)$ with a set of five tasks fitting into the box.}
\label{fig:box}
\end{figure}

We use the next lemma to show that there is a set of constantly many
boxes $B_{w,d}$ in which almost all tasks in $T_{w,d}\cap\overline \OPT^{(M)}_r$
fit, i.e., up to a factor of $1+\eps$ in the profit. Also, the
capacity used by these boxes is at most the capacity used by the tasks
in $T_{w,d}\cap\overline \OPT^{(M)}_r$ on each edge. Formally, we apply the
following lemma to the set $T_{w,d}\cap\overline \OPT^{(M)}_r$ and obtain a set of at
most $1/\eps^{2}$ boxes that we denote by $B_{w,d}$. 
We prove this lemma with a similar linear grouping scheme as in \cite{vl81} (see also \cite{williamson2011design})
to show that we can approximate the demand profile of such a set of tasks by a simpler profile with a constant number of steps only.

\begin{restatable}{lemma}{boxPacking}
\label{lem:boxPacking} 
Let $T'$ be a set of tasks, all having identical weight $w$ and identical demand $d$, and a schedule $P(i) \subseteq E$ for every $i\in T'$ such that there is an edge~$f$ that is contained in each path $P(i)$.
Then, there exists a set $S\subseteq T'$, a
set of boxes $B$ with $|B|\le1/\eps^{2}$, and a partition of $S$
into sets $\left\{ S_{b}\right\} _{b\in B}$ such that:
\begin{enumerate}[label=(\arabic*)]

\item \label{boxes1-1} $\sum_{b\in B:e\in\pb(b)}h(b)\le\sum_{i\in T':e\in P(i)}d(i)$ for each $e\in E$,

\item \label{boxes2-1} $w(S)\geq(1-O(\eps))w(T')$,

\item \label{boxes3-1}  the tasks in $S_{b}$
fit into $b$ for each $b\in B$.
\end{enumerate}
\end{restatable}

\begin{proof}
If $|T'|\le1/\eps^{2}$, then we can set the boxes to coincide with
the embedded tasks in $T'$, so assume $|T'| > 1/\eps^{2}$.
First, we sort the tasks
$T'$ in increasing order of the \fab{leftmost edge of $P(i)$}, and partition
them into $1/\eps$ sets $S_{1},\ldots,S_{1/\eps}$ such that $\lceil\eps|T'|\rceil \geq |S_r| \geq |S_{r+1}| \geq \lfloor\eps|T'|\rfloor$ for $r\in \{1, \dots, 1/\varepsilon-1\}$.
For every $r\in \{1, \dots, 1/\varepsilon\}$, let $\mu_{r}$ be the left-most edge \fab{in $P(i)$} among all the tasks in $S_r$.

For each $r\in\{1, \dots, 1/\eps\}$, we sort the tasks in
$S_{r}$ in decreasing order of the right-most edge \fab{in $P(i)$}, and partition
them into $1/\eps$ sets $S_{r,1},\dots,S_{r,1/\eps}$ such that $\lceil\eps|S_r|\rceil \geq |S_{r,q}| \geq |S_{r,q+1}| \geq \lfloor\eps|S_r|\rfloor$ for $q\in \{1, \dots, 1/\varepsilon-1\}$.
For $r, q \in \{1, \dots, 1/\varepsilon\}$, let $\nu_{r,q}$ be the right-most edge \fab{in $P(i)$} among all the tasks in $S_{r,q}$.
Notice that for any task in $i\in S_{r,q}$, $P(i)$ lies on the path from $\mu_r$ to $\nu_{r,q}$.

For each $r\in\{2, \dots, 1/\eps\}$ and each $q\in\{2,\dots,1/\eps\}$, we
create a box $b_{r,q}$ with height $d\cdot|S_{r,q}|$ and $\pb(b_{r,q})$ being the path
from $\mu_{\fab{r}}$ to $\nu_{\fab{r,q}}$. We define $S=\bigcup_{r,q\in\{2,\dots,1/\eps\}} S_{r,q}$.
We show that the boxes and set $S$ defined in this way satisfy the
lemma.

\ref{boxes1-1}:
Assume first that $e$ lies between $\mu_r$ and $\mu_{r+1}$ or $e=\mu_r$ for some $r \in \{1, \dots, 1/\varepsilon-1\}$. Then
\begin{equation*}
    \sum_{b\in B:e\in \pb(b)} h(b)
    \leq \sum_{2 \leq r' \leq r, 2 \leq q \leq 1/\eps} h(b_{r',q})
    \leq \sum_{2 \leq r' \leq r} d\cdot |S_{r'}|
    \leq \sum_{i\in T': e\in P(i)} d(i).
\end{equation*}
Similarly, if now $e$ lies between $\nu_{r,q+1}$ and $\nu_{r,q}$ or $e=\nu_{r,q}$ for some fixed $r$, then
\begin{equation*}
    \sum_{q':e\in \pb(b_{r,q'})} h(b_{r,q'})
    = \sum_{2 \leq q'\leq q} d\cdot |S_{r,q'}|
    \leq \sum_{i \in S_r: e\in P(i)} d(i). 
\end{equation*}
Summing up the \fab{above} terms over $r \in \{1, \dots, 1/\varepsilon\}$ proves the claim. If finally $e$ lies between $\mu_{1/\varepsilon}$ and the leftmost $\nu_{r,q}$ we have
\begin{equation*}
    \sum_{b\in B:e\in \pb(b)} h(b)
    = \sum_{b\in B} h(b)
    = d\cdot |S|
    \leq \sum_{i\in T'} d(i)
    = \sum_{i\in T':e\in P(i)} d(i).
\end{equation*}
\ref{boxes2-1}: Clearly, $T'\setminus S=S_{1} \cup (\bigcup_{r\in\{2,\dots,1/\eps\}}S_{r,1})$.
Since $T' > 1/\eps^2$, we have\fabr{I put $\leq$ since in the sum some inner ceilings might be floors, right? \alex{correct}}
\begin{equation*}
    |T' \setminus S|
    \fab{\leq} \lceil\eps|T'|\rceil + \sum_{2 \leq r \leq 1/\eps} \lceil\eps\lceil\eps|T'|\rceil\rceil
    \leq 2\eps |T'| + \frac{1}{\eps} \cdot 2\eps \cdot 2\eps |T'| = O(\eps)|T'|.
\end{equation*}
\ref{boxes3-1}: Follows by construction and the definition of the box.
\end{proof}

We guess the at most $1/\eps^{2}$ boxes $B_{w,d}$. 
For each box $b\in B_{w,d}$, let $S_{b}\subseteq T_{w,d}\cap\overline \OPT^{(M)}_r$
be the subset of $T_{w,d}\cap\overline \OPT^{(M)}_r$ that is assigned to $b$
when applying Lemma~\ref{lem:boxPacking} to $T_{w,d}\cap\overline \OPT^{(M)}_r$;
we guess $|S_{b}|=:n_{b}$. We do this procedure for the set $T_{w,d}$
for each combination of $w$ and $d$. 
Let $B_{0}$ denote the set of all boxes that we guessed this way.

Let $I_{r}^{L}$ and $I_{r}^{R}$ be the two children of $I_{r}$.
We recurse on a \emph{left subproblem }corresponding to $I_{r}^{L}$
and on a \emph{right subproblem }corresponding to $I_{r}^{R}$. In
the right subproblem, we are given as input the set 
consisting of each task $i\in T$
such that $\tw(i)\subseteq I_{r}^{R}$. 
Note that such a task must have level $1$ or larger, so the tasks of level $0$ are not considered in the subproblem corresponding to $I_r^R$.
The capacity of each edge $e\in I_{r}^{R}$ is defined as $u'(e)\coloneqq u(e)-\sum_{b\in B_{0}:e\in\pb(b)}h(b)$,
that is, the boxes in $B_{0}$ use a certain amount of the capacities
of the edges that we cannot use.  The goal is to compute a feasible solution
to this smaller instance, i.e., a set of tasks $Q^{R}$ with a path
$P(i)\subseteq \tw(i)\subseteq I_{r}^{R}$ for each task $i\in Q^{R}$ that respects the edge capacities $u'(e)$. The objective is to maximize $w(Q^{R})$.

In the left subproblem, we are given as input the set of all tasks
$i\in T$ such that $\tw(i)\subseteq I_{r}^{L}$ or $\ell(i)=0$;
let $T^{L}$ denote this set of tasks. The capacity of each edge $e\in I_{r}^{L}$
is defined as $u'(e)\coloneqq u(e)-\sum_{b\in B_{0}:e\in\pb(b)}h(b)$ analogously to the right subproblem. Also, we are
given the set of boxes $B_{0}$ and a value $n_{b}$ for each box
$b$. The goal is to select a subset $Q\subseteq T^{L}$, a partition
of $Q$ into a set $Q^{L}$ and a set $Q_{b}$ for each box $b\in B_{0}$,
and a path $P(i)$ for each task $i\in Q$
such that
\begin{itemize}
\item the set $Q^{L}$ with the path $P(i)\subseteq\tw(i)\cap I_{r}^{L}$ for each task $i\in Q^{L}$
forms a feasible solution for the given edge capacities $u'(e)$, i.e., for each edge $e\in I^L_r$ we have $\sum_{i\in Q^{L}:e\in P(i)}d(i)\leq u'(e)$, and
\item for each box $b\in B_{0}$ we have that $|Q_{b}|=n_{b}$, $Q_{b}$
fits into $b$, and $P(i)\subseteq \tw(i)\cap \pb(b)$ for each $i\in Q_b$.
\end{itemize}
Observe that the given boxes $B_{0}$ do not interact at all with
the edges $I_{r}^{L}$ and their given capacities. However, the subpath
$\pb(b)$ for a box $b\in B_{0}$ is important since a task $i$ fits
into a box $b$ only if $\left|\tw(i)\cap\pb(b)\right|\ge p(i)$.

Given a solution to the left and the right subproblems, we combine
them to a solution to the overall problem in the obvious way: the solution is given by the tasks $Q=Q^{L}\cup Q^R\cup (\cup_{b\in B_0}Q_b)$ with the respective paths $P(i)$. Each of our quasi-polynomially
many guesses yields a candidate solution; among all feasible candidate
solutions, we output the solution with maximum profit $w(Q)$.

\subsection{Arbitrary subproblems}

In the following, we describe our routine for arbitrary subproblems.
The reader may think of the subproblem for the left child of the root
$I_{r}$, i.e., the interval $I_{r}^{L}$ (which already incorporates all the challenges of a generic subproblem). We assume that the input
consists of
\begin{itemize}
\item an interval $I\in\I$ of level $\ell$ with a capacity $u'(e)$ for each edge $e\in I$,
\item a set of tasks $T'$, and
\item a set of boxes $B$ and a value $n_{b}$ for each box $b\in B$.
\end{itemize}
We remark that, for a box $b\in B$, not necessarily $\pb(B)\subseteq I$. Furthermore, given a task $i\in T'$, it might happen that $\tw(i)\not\subseteq I$; in this case, we will ensure that $\tw(i)$ contains the rightmost edge of $I$ by construction.
We remark that it might happen that $i\not\in T$, namely $i$ is not one of the input task. This happens because during the process, we introduce some \emph{artificial tasks} whose role will become clear later. Like the regular tasks, each artificial task $i$ has a weight $w(i)$, a demand $d(i)$, a length $p(i)$, and a time window~$\tw(i)$.

Let $I^{L}$ and $I^{R}$ denote the two children of $I$ and let
$\mi(I)$ denote the middle vertex of $I$. The goal is to compute
a set of tasks $Q\subseteq T'$, a partition of $Q$ into sets $Q_{I}$
and a set $Q_{b}$ for each box $b\in B$, and a path $P(i)$
for each task $i\in Q$ such that
\begin{itemize}
\item the tasks in $Q_{I}$ with their paths $P(i) \subseteq\tw(i)\cap I$ obey the edge capacities
of $I$, i.e., for each edge $e\in I$ we have $\sum_{i\in Q_{I}:e\in P(i)}d(i)\leq u'(e)$,
\item for each box $b\in B$ the tasks in $Q_{b}$ fit into $B$, $|Q_{b}|=n_{b}$, and $P(i)\subseteq \tw(i)\cap \pb(b)$ for each $i\in Q_b$.
\end{itemize}
The objective is to maximize $w(Q_{I})$, i.e., the weight of the
tasks in $Q_{I}$. Note that the tasks assigned to the boxes $B$
do not yield any profit in this subproblem. Intuitively, we accounted
the profit of such tasks already in subproblems of higher levels in
the recursion.

Let $(\overline{\OPT},P^*(\cdot))$ denote an optimal left constrained solution to the given subproblem.
We denote by $\overline{\OPT}_{I}$ and $\left\{ \overline{\OPT}_{b}\right\} _{b\in B}$
the corresponding partition of $\overline{\OPT}$ (i.e., $\overline{\OPT}_{I}$ consists of the tasks in $\overline{\OPT}$ scheduled on $I$ and $\overline{\OPT}_{b}$ contains the task in $\overline{\OPT}$ that are placed inside $b$).
The base cases of our recursion are there subproblems where $|I|=1$. 
In this case, we can solve the subproblem exactly in quasi-polynomial time. We guess for
each combination of a weight $w$ and a demand $d$ how many tasks
with this weight and demand are contained in $\overline{\OPT}_{I}$.
Then, the remaining problem can be reduced easily to an instance of
$b$-matching.
\begin{restatable}{lemma}{baseCase}\label{lem:base-case}
Any subproblem with $|I|=1$ can be solved exactly in $n^{O(WD)}$ time, where $W$ and $D$ denote the number of distinct weights and demands, respectively. 
\end{restatable}

\begin{proof}
    Consider\fabr{rewrote this proof. Pls check in particular if what I said about $\overline{\OPT}_e$ is correct} a subproblem on an interval $I =\{e\}$, $e\in E$, with boxes $B$ (with an $n_b$ associated to each $b\in B$). Notice that, by definition, in an optimal solution $(\overline{\OPT},\overline{P}(\cdot))$ for this subproblem, the tasks $\overline{\OPT}_{e}\subseteq \overline{\OPT}$ not assigned to any box must be scheduled with $\overline{P}(i)=\{e\}$ (in particular they must have $p(i)=1$). For each pair $(w,d)$ of weight and demand, we guess how many tasks $n_{w,d}$ of that weight and demand are contained in $\overline{\OPT}_{e}$. Notice that the number of possible guesses is at most $n^{O(\alex{WD)}}$.
    We create a corresponding box $b$ with $\pb(b)=\{e\}$, $w(b)=w$, $d(b)=d$, and $h(b)=n_b\cdot d$, where $n_b=n_{w,d}$, and add $b$ to the set of boxes $B$. Given the current guess, we check if it is possible to assign precisely $n_b$ tasks to each $b\in B$, and call it a feasible guess if this happens (we discuss later how to check it). We return the solution corresponding to the feasible guess of largest profit. Here the paths $\overline{P}(i)$ are chosen in any feasible way \ant{(inside the boxes the tasks are assigned to)}. 
    
    In order to check the feasibility as mentioned before, we build a bipartite graph where on the left we place a node $v_i$ for each task $i$ and on the right a node $v_b$ for each box $b$. We add an edge $\{v_i,v_b\}$ whenever task $i$ can be assigned to box $b$. In this graph we check whether there exists a $B$-matching $M$ where at most $1$ edge of $M$ is incident on each node $v_i$ and exactly $n_b$ edges of $M$ are incident to $n_b$. Checking the existence of such a $B$-matching can be done in polynomial time by standard matching theory (see, e.g., \cite[Theorem 21.9]{Schrijver03}).
\end{proof}

Assume now that $|I|>1$. For each task $i\in\overline{\OPT}_{I}$ there are three possibilities:
\begin{enumerate}
\item\label{case:scheduledLeft} $P^{*}(i)$ is contained in $I^{L}$, or
\item\label{case:scheduledMiddle} $\mi(I)$ is in the interior of $P^{*}(i)$, i.e., $\mi(I)$ is a
vertex of $P^{*}(i)$ but it is neither the leftmost nor the rightmost
vertex of $P^{*}(i)$, or
\item\label{case:scheduledRight} $P^{*}(i)$ is contained in $I^{R}$.
\end{enumerate}
Let $\overline{\OPT}^{(L)},\overline{\OPT}^{(M)},$ and $\overline{\OPT}^{(R)}$
denote the tasks in $\overline{\OPT}_{I}$ for which case \ref{case:scheduledLeft}, \ref{case:scheduledMiddle} and \ref{case:scheduledRight} applies, respectively.
Note that the partition of $\overline{\OPT}$ into 
$\overline{\OPT}^{(L)},\overline{\OPT}^{(M)},$ and $\overline{\OPT}^{(R)}$ is different from the partition of 
$\OPT$ into $\OPT^{(L)},\OPT^{(M)},$ and $\OPT^{(R)}$, 
since the former is with respect to $\mi(I)$ for each task $i\in \overline{\OPT}$ of \emph{any} level.
Recall also that $\overline{\OPT}$ is left constraint and, hence,
assuming that $I$ is an interval of level $\ell$,
for no task $i$ of level $\ell$, its scheduled path $P^*(I)$
in the optimal solution
is contained in $\overline{\OPT}^{(R)}$.
More precisely, for each task $i \in \overline \OPT^{(R)}$, it must be the case that either $\tw(i) \subseteq I^R$ (so $i$ is of level $\ell+1$ or deeper), $I^R \subseteq \tw(i)$ (so $i$ is of level $\ell-1$ or higher) or the leftmost edge of $\tw(i)$ is in $I^R$ but not all edges of  $\tw(i)$
(so $i$ is of level $\ell-1$ or higher).

First, we guess boxes for the tasks in $\overline{\OPT}^{(M)}$. Formally,
for each combination of a weight $w$ and demand $d$, we apply Lemma~\ref{lem:boxPacking}
to the set $\{i\in\overline{\OPT}^{(M)}:w(i)=w\wedge d(i)=d\}$. Let
$\bar{B}_{w,d}$ denote the resulting set of boxes and for each box
$b\in\bar{B}_{w,d}$ let $\overline{\OPT}_{b}^{(M)}$ be the resulting
set of tasks. 
We guess the boxes in $\bar{B}_{w,d}$ and $|\overline{\OPT}_{b}^{(M)}|=:n_{b}$ for each
box $b\in\bar{B}_{w,d}$.
Let $\bar{B}$ denote the union of all sets of boxes $\bar{B}_{w,d}$ that
we guessed in this way.

Let $\overleftarrow{T}'\subseteq T'$ denote all tasks in $i\in T'$
such that $I^{R}\subseteq\tw(i)$ and $I^{L}\cap\tw(i)\ne\emptyset$, i.e., task whose  time windows stick into $I$ from the
right and intersect $I^L$. 
For the subproblem for $I_{r}^{L}$, the reader may 
imagine that $\overleftarrow{T}'$ are all tasks in $T_{0}$ 
whose time windows start in the left child of $I_{r}^{L}$.
For a task $i\in\overleftarrow{T}'\cap\overline{\OPT}_{I}$
it is possible that $i\in\overline{\OPT}^{(L)}$ or that $i\in\overline{\OPT}^{(R)}$.
Therefore, at first glance it would make sense to pass $i$ to our
subproblem for $I^{L}$ and to our subproblem for $I^{R}$. However,
we must avoid that our subproblem for $I^{L}$ and our subproblem
for $I^{R}$ both select $i$. On the other hand, we do not want to
sacrifice a factor of 2 by, e.g., omitting $\overline{\OPT}^{(L)}$
or $\overline{\OPT}^{(R)}$. For example, if $\overleftarrow{T}'=T_{0}$
in the subproblem for $I_{r}^{L}$, then this would lose a second factor
of 2 on the profit of the tasks in $T_{0}$, while we want to bound
our approximation ratio by $2+O(\eps)$.

We execute the following steps instead:
\begin{enumerate}[label=\Roman*.]
\item Perform a linear grouping step in which we round up the task lengths so that there are at most polylogarithmically many distinct lengths in $\overleftarrow{T}'\cap\overline{\OPT}^{(R)}$;
\item For each combination of a demand, weight and (rounded) length, guess
how many such tasks are contained in $\overleftarrow{T}'\cap\overline{\OPT}^{(R)}$;
\item Give a suitable number of these rounded tasks as \emph{artificial input tasks} to our right subproblem for $I^{R}$;
\item Enforce our left subproblem for $I^{L}$ to leave sufficiently
many tasks from $\overleftarrow{T}'$ unassigned such that we can schedule these tasks later in the places in which the solution to the right subproblem schedules the  artificial tasks (we model this requirement for $I^{L}$ via suitable additional box constraints).
\end{enumerate}
    Formally, for each combination of a weight $w$ and demand $d$ we  apply the following lemma to the set $\overleftarrow{T}'_{w,d}=\big\{i \in \overleftarrow{T}'\cap\overline{\OPT}^{(R)}:w(i)=w\wedge d(i)=d\big\}$.

\begin{restatable}{lemma}{harmonicGrouping}\label{lem:harmonic-grouping}
Let $T'$ be a set of tasks,
all having identical weight $w$ and identical demand $d$,
and a schedule $P'(i) \subseteq I \subseteq \tw(i)$ for every $i\in T'$.
Then, there exists a set of \emph{artificial tasks}
$\tilde{T}$ (not contained in $T$), a set of boxes $\tilde{B}$
and a value $n_{b}\in\N$ for each box $b\in\tilde{B}$ with $\sum_{b\in \tilde B} n_b = |\tilde T|$ such that
\begin{enumerate}[label=(\arabic*)]
\item $(1-\eps)|T'|\le|\tilde{T}|\le|T'|$, and 
the tasks in $\tilde{T}$ have at most $1/\eps$ distinct lengths,
\item $\tw(i)=I$, $d(i)=d$, and $w(i)=w$ for each task $i\in\tilde{T}$,
\item there is a solution $(\tilde{T},\tilde{P}(\cdot))$ for $\tilde{T}$ such that
$\sum_{i\in\tilde{T}:e\in\tilde{P}(i)}d(i)\le\sum_{i\in T':e\in P'(i)}d(i)$
for each edge $e\in I$, 
\item 
there is a collection $\text{\ensuremath{\left\{ Q'_{b}\right\} }}_{b\in\tilde{B}}$ of pairwise disjoint subsets of $T'$ 
such that for each $b \in \tilde{B}$ the set $Q'_b$ fits into $b$
and $|Q'_b|=n_b$,
\item given a combination of a solution $(\tilde{S},\tilde{P}(\cdot))$
		for a subset $\tilde{S}\subseteq \tilde{T}$ and any collection of pairwise disjoint sets of tasks $\{ Q_{b}\}_{b\in\tilde{B}}$,
		in which for each box $b\in\tilde{B}$, the set $Q_{b}$ fits into~$b$, $|Q_{b}|=n_b$, and each task $i \in Q_b$ has a time window containing the leftmost edge of $I$; then there is an injective function $f:\tilde{S} \to \bigcup_{b\in\tilde{B}}Q_{b}$ such that $f(i)$ can be scheduled instead of $i$, i.e.,$|\tilde{P}(i)\cap \tw(f(i))|\geq p(f(i))$, for each $i\in \tilde{S}$.
\end{enumerate}
\end{restatable}

\begin{proof}
    Assume first $|T'| > 1/\eps$ and therefore $\lfloor\eps|T'|\rfloor \geq 1$.
    Sort the tasks in $T'$ in order of decreasing length (breaking ties arbitrarily) and partition them into $1/\eps$ sets $T'_{1},\ldots,T'_{1/\eps}$ such that $\lceil\eps|T'|\rceil \geq |T'_r| \geq |T'_{r+1}| \geq \lfloor\eps|T'|\rfloor$ for $r\in [1/\varepsilon-1]$.
    Let $p_{r} = \max_{i\in T'_r} p(i)$ for $r\in [1/\eps]$ and let $\tilde T_{r}$ consist of $|T'_r|$ artificial tasks, each with demand $d$, weight $w$, length $ p_{r}$ and time window $I$ for $r \in \{2, \dots, 1/\eps\}$. Let $\tilde T = \bigcup_{2 \leq r \leq 1/\eps} \tilde T_r$.
    Notice that $|\tilde T| = |T' \setminus T'_{1/\eps}| = |T'| - \lfloor\eps|T'|\rfloor \fab{\geq (1-\eps)} |T'|$.
    The first two claims of the lemma follow.
    
    Sort the tasks in $\tilde T$ in order of decreasing length, and let $g$ be the map that maps the first element of $\tilde T$ to the first element of $T'$, the second element of $\tilde T$ to the second element of $T'$ and so on.
    For $i\in \tilde T_r$, we have $g(i) \in T'_{r-1} \cup T'_{r-2}$, we have $p(i) =  p_{r} \leq \min_{j \in T'_{r-1} \cup T'_{r-2}}p_{j} \leq p(g(i))$.
    We construct the schedule $\tilde P(\cdot)$ of $\tilde T$ as follows: for every $i \in \tilde T_r$, $r \in \{2, \dots, 1/\eps\}$, let $\tilde P(i)$ be the interval of length $p_r=p(i)$ whose rightmost edge is the rightmost edge of $P'(g(i))$, thus $\tilde P(i) \subseteq P'(g(i))$.
    This means that for every $i \in \tilde T$ and $e\in I$, we obtain $d(i) = 0 \leq d(g(i))$ if $e\notin \tilde P(i)$ and $d(i) = d(g(i))$ otherwise.
    The third claim follows by summing up over $\sum_{2\leq r\leq 1/\eps} \sum_{j\in T'_r}$.

    We define the boxes in $\tilde B$ as follows: for every $r \in \{2, \dots, 1/\eps\}$, we define the box $b_r$ to have height $d \cdot |T'_r|$, demand $d$, weight $w$ and $\pb(b_r)$ being the rightmost subpath in $I$ of length $p_r$; let also $n_{b_r} = |T'_r|$.
    By construction, for every $r \in \{2, \dots, 1/\eps\}$, $T'_r$ fits in $b_r$. The fourth claim of the lemma follows by setting $Q'_{b_r} = T'_r$.

    We now describe the function $f$ in the fifth claim of the lemma.
    W.l.o.g.\ let $\tilde S = \tilde T$ since otherwise just take the restriction of $f$ to $\tilde S$.
    For every $r \in \{2, \dots, 1/\eps\}$, let $f_r$ be an arbitrary bijection from $\tilde T_r$ to $Q_{b_r}$ (which exists since $|\tilde T_r| = n_{b_r} = |Q_{b_r}|$) and $f$ be the union of the $f_r$'s (i.e., $f(i) = f_r(i)$ if $i \in \tilde T_r$).
    By construction, for $i \in \tilde T_r$, we have $p(i) = p_r = |\pb(b_r)| \geq p(f(i))$.
    \alex{As a task $f(i)\in Q_{b_r}$ can be placed in $b_r$ we obtain $|\tw(f(i))\cap \pb(b_r)|\geq p(f(i))$. As $\pb(b_r)$ consists of the $p_r$ rightmost edges of $I$ and $\tw(f(i))$ also contains the leftmost edge from $I$, it follows that $\tw(f(i))$ contains at least all edges from $I$ except the $p_r-p(f(i))$ rightmost ones. Thus $|\tilde P(i) \cap \tw(f(i))|=|\tilde{P(i)}|-|\tilde P(i) \setminus \tw(f(i))|\geq |\tilde{P}(i)|-|I \setminus \tw(f(i))|\geq p_r-(p_r-p(f(i)))=p(f(i))$.}
    
    If $|T'| \leq 1/\eps$, we can define $\tilde T$ as a copy of $T'$: for each task $i \in T'$ we add a task $i'$ to $\tilde T$ with weight $w$, demand $d$, length $p(i)$ and time window $I$.
    The first two claims trivially hold.
    For every $i \in T'$, we set $\tilde P(i') = P'(i)$ and create a box $b_{i'}$ with height $d$, demand $d$, weight $w$ and $\pb(b_{i'})$ being the rightmost subpath in $I$ of length $p(i)$; let also $n_{b_{i'}} = 1$. It follows that for any $b = b_{i'} \in \tilde B$, we can choose $Q'_{b} = \{i\}$.
    Notice now that any set $Q_b = Q_{b_{i'}}$ is a singleton $\{j\}$.
    The function $f$ is thus defined as $f(i') = j$ where $Q_{b_{i'}} = \{j\}$.
    The last three claims of the lemma follow from the same argument as above.
\end{proof}

We sketch now how we will apply
Lemma~\ref{lem:harmonic-grouping}. Since we are unable to guess each set $\overleftarrow{T}'_{w,d}$ at this stage, we rather guess the corresponding set of artificial tasks $\tilde{T}_{w,d}$ according to Lemma \ref{lem:harmonic-grouping}. Notice that this is doable in polynomial time (per pair $(w,d)$) since it is sufficient to guess the distinct lengths $p_{1},...,p_{1/\eps}$ and the respective number of tasks $\tilde{n}_{1},...,\tilde{n}_{1/\eps}$. We will pass on the tasks $\tilde{T}_{w,d}$ to our right subproblem as placeholders: intuitively, they will reserve some capacity on the right subproblem which will eventually be occupied by actual tasks (potentially the tasks $\overleftarrow{T}'_{w,d}$) computed in the left subproblem: this replacement exploits Property 5 of the lemma.
Also, we guess the boxes $\tilde{B}_{w,d}$ and the number $n_b \in \NN$ for each $b\in \tilde B_{w,d}$ corresponding to $\tilde{T}_{w,d}$, which also takes polynomial time. These boxes will be passed on to the left subproblem: intuitively, the tasks placed in these boxes in the left subproblem will replace the artificial tasks in $\tilde{T}_{w,d}$.

Let $\tilde{T}$ denote the union of all sets $\tilde{T}_{w,d}$
guessed in this way and let $\tilde{B}$ denote the union of all the sets of boxes
$\tilde{B}_{w,d}$ guessed in this way. We will call the tasks in $\tilde{T}$ \emph{artificial tasks.}

We will recurse on a left and a right subproblem, corresponding to
$I^{L}$ and $I^{R}$, respectively. For each box $b\in{B}$,
we intuitively guess how many input tasks from the left and the right
subproblem are assigned to $b$ in $\overline{\OPT}$. Formally, we
guess values $n_{b}^{L},n_{b}^{R}\in\N_{0}$ where $n_{b}^{L}$ is
the number of tasks $i\in\overline{\OPT}_{b}$ for which $\tw(i)\cap I^{L}\ne\emptyset$
and $n_{b}^{R}$ is the number of tasks $i\in\overline{\OPT}_{b}$
for which $\tw(i)\cap I^{L}=\emptyset$ (in particular, $n^L_b + n^R_b = n_b$). The input for the left subproblem
consists of
\begin{itemize}
\item the interval $I^{L}$, where each edge $e\in I^{L}$ has capacity
 $u'(e)-\sum_{b\in\bar{B}:e\in\pb(b)}h(b)$,
\item the task set $\left\{ i\in T':\tw(i)\cap I^{L}\ne\emptyset\right\} $,
\item the boxes $B\cup\bar{B}\cup \tilde{B}$, the value $n_{b}$ for each
box $b\in \bar B\cup\tilde{B}$, and the value $n_{b}^{L}$ for each box
$b\in B$.
\end{itemize}
The input for the right subproblem consists of
\begin{itemize}
\item the interval $I^{R}$, where each edge $e\in I^{R}$ has capacity $u'(e)-\sum_{b\in\bar{B}:e\in\pb(b)}h(b)$,
\item the task set $\left\{ i\in T':\tw(i)\cap I^{L}=\emptyset\right\} \cup\tilde{T}$,
\item the boxes $B$ and the value $n_{b}^{R}$ for each box $b\in B$.
\end{itemize}

Suppose that recursively we computed a solution to the left and the
right subproblem. We combine them to a solution to the (given) subproblem
corresponding to $I$ as follows. Let $Q_{I^{L}},\left\{ Q_{b}\right\} _{b\in B\cup\bar{B}\cup\tilde{B}}$
denote the computed tasks in the left subproblem, and let $Q'_{I^{R}},\left\{ Q'_{b}\right\} _{b\in B}$
denote the computed tasks in the right subproblem. For each task $i$
in these computed sets of tasks, the returned solutions include a
computed path $P(i)$. We compute the injective function $f$ due
to the last property of Lemma~\ref{lem:harmonic-grouping}
by solving a $b$-matching instance. If such a function $f$
does not exist, we simply discard the considered combination
of guesses. 

We want to return a solution that consists of the tasks
in $Q_{I^{L}}\cup\left(\bigcup_{b\in\bar{B}}Q_{b}\right)\cup(Q_{I^{R}}\setminus\tilde{T})$
and additionally a set of tasks $\tilde{Q}$ that we use to replace
the artificial tasks in $\tilde{T}$ in the solution $Q_{I^{R}}$.
For each task $i\in Q_{I^{L}}\cup\left(\bigcup_{b\in\bar{B}}Q_{b}\right)\cup(Q_{I^{R}}\setminus\tilde{T})$
we keep the path $P(i)$ that we obtained from the solution of the left or right
subproblem, respectively. We want to replace the tasks in $Q_{I^{R}}\cap\tilde{T}$ and, to this end, we 
compute the set of tasks $\tilde{Q}\coloneqq f(Q_{I^{R}}\cap\tilde{T})\subseteq\bigcup_{b\in\tilde{B}}Q_{b}$
and define paths for them as follows. For each
task $i\in Q_{I^{R}}\cap\tilde{T}$ we schedule the task $f(i)\in \tilde{Q}$
within the edges of $P(i)$, i.e., we compute a path $P(f(i))$
for $f(i)$ such that $P(f(i))\subseteq P(i)$. This
is possible since $|P(i)\cap\tw(f(i))|\ge p(f(i))$ due to the last property
of 
Lemma~\ref{lem:harmonic-grouping}.
Notice that this replacement does not violate the edge
capacities since we omit 
$i\in\tilde{T}$ from our solution and $d(i)=d(f(i))$. 

This yields a candidate solution consisting of
the tasks
$Q_{I}\coloneqq Q_{I^{L}}\cup\left(\bigcup_{b\in\bar{B}}Q_{b}\right)\cup(Q_{I^{R}}\setminus\tilde{T})\cup\tilde{Q}$ and
for each box $b\in B$ the tasks in $Q_{b}\cup Q'_{b}$
and their respective computed paths $P(i)$.
Each combination of our quasi-polynomial number of guesses yields
a candidate solution. We reject a candidate if the resulting solution
is infeasible; among the remaining candidate solution, we return the
solution that maximizes $w(Q_{I})$. Recall here that $\tilde{Q}$
are real tasks that we used to replace the artificial tasks $\tilde{T}$; we count the profit of the tasks in $\tilde{Q}$ but
not the profit of $\tilde{T}$. Also, we do not consider the profit
of the tasks in $B$.

\subsection{Running time}

\ant{We bound the running time of the algorithm.}

\begin{restatable}{lemma}{runningTime}
\label{lem:runningTime}
\begin{sloppypar}
Let $W$ be the number of different profits, and let $D$ be the number of different demands of the tasks in our \twUFP{} instance.
Our algorithm runs in time $O((mn)^{O(WD\log^2 m )/\eps^2})$
and its returned solution is always feasible.
\end{sloppypar}
\end{restatable}


\ant{We first need the following lemma to bound the number of ``guessable'' boxes at each stage of the algorithm.}

\begin{restatable}{lemma}{numberOfBoxes}
\label{lem:numberOfBoxes}
There are at most $\alex{m^2 \cdot n}$ feasible boxes with a given demand $d$ and profit $w$.
\end{restatable}
\begin{proof}
    Each box is specified by a path, i.e., its left-most and right-most edge, by a height, demand and weight. \alex{There are at most $m$ possibilities for the left-most and $m$ possibilities for the right-most edge and the height is an integer multiple of $d$ and at most $n\cdot d$, thus there are $n$ possibilities for the height.}
    Thus, there are at most $m^2 \cdot n$ boxes with a given demand and profit that have capacity not exceeding the maximum capacity on the path. 
\end{proof}

\begin{proof}[\ant{Proof of Lemma~\ref{lem:runningTime}}]
The total running time of the algorithm is upper bounded by $K \cdot T \cdot \poly(mn)$, where $K$ is the total number of recursive calls and T is the running time needed to solve the base case\fab{s}. 
Let us first bound the number of guesses we need to check in each recursive call. Recall that in each subproblem, we guess boxes $\bar B$ with a number $n_b$ for each $b \in \bar B$, the values $p_1, \dots, p_{1/\eps}$ and numbers $n_1, \dots, n_{1/\eps}$, the boxes $\tilde B$ and a number $n_b$ for $b\in \tilde B$, and for each $b\in B$ we guess two numbers $n^L_b$ and $n^R_b$, \fab{with $n_b=n^L_b+n^R_b$}.

Note that at each recursive call, the number of boxes for each \fab{pair} $(w,d)$ increases by at most $2/\eps\alex{^2}$. (At most $1/\eps\alex{^2}$ for $\bar B_{w,d}$ and at most $1/\eps$ for $\tilde B_{w,d}$.) Thus at every recursive call there are at most $2/\eps\alex{^2} \cdot \log_2 m$ different boxes with a given demand and profit. 

For a given subproblem and for a given \fab{pair} $(w,d)$: there are at most $m^{2/\eps^2}  n^{1/\eps^{2}}$ possible values for the set of boxes $\bar B_{w,d}$ \ant{by Lemma~\ref{lem:numberOfBoxes}} and for each possible value we pick one number $n_b \in [n]$ representing the number of tasks, i.e., $m^{2/\eps^2}  n^{1 + 1/\eps^{2}}$ \fab{combinations in total}. 
\edi{Similarly, for $\tilde B$ and the corresponding numbers, there are at most $m^{2/\eps} n^{1+1/\eps}$ possible values for the combination of a set of boxes $\tilde B_{w,d}$ and a number $n_b \in [n]$.} 
There are at most $n^{1/\eps}$ possible tuples for values $p_{j}$ and at most $n^{1/\eps}$ possible numbers $n_{j}$.
Since the number of boxes is bounded by $2/\eps\alex{^2} \cdot  \log_2 m$ at every subproblem, there are at most $n^{2/\eps\alex{^2} \cdot \log_2 m}$ choices for $n^L_b$ \edi{(recall that $n^R_b = n_b - n^L_b$).
In total, the number of guesses need for each \fab{pair} $(w,d)$ at each recursive call is bounded by $m^{2/\eps^2 + 2/\eps} \cdot n^{2+ 1/\eps^2 + 1/\eps} \cdot n^{2/\eps} \cdot n^{2/\eps^2 \cdot \log_2 m} \le \left(mn\right)^{8/\eps^2 \cdot \log_2 m}$ (assuming $\eps \le 1$ and $\log_2 m \ge 1/\eps$). 
Thus, the total number of choices considered over all \fab{pairs} $(w,d)$  in any subproblem is $O((mn)^{8 WD\log m/\eps^2})$. Since there are at most $\log m$ recursive levels, one has $K=O((mn)^{O(WD\log^2 m)/\eps^2})$.
Since $T = n^{O\left(WD \right)}$ we obtain the lemma.}
\end{proof}

\subsection{Approximation guarantee}

\ant{We bound the approximation ratio of our algorithm.}

\begin{restatable}{lemma}{approximationRatio}
\label{lem:approximation-ratio}The approximation ratio of our algorithm
is bounded by $2+O(\eps\log m)$.
\end{restatable}

To prove the lemma, we show that the output of the algorithm is within factor $(1+\eps \log m)$ of the optimal \emph{left constrained} solution. 
Recall, that a solution is left constrained if for each task $i\in T$ of level $\ell$ and the unique interval $I$ of level $\ell$ containing it, then $i$ is allowed to be scheduled either on the left of $\mi(I)$ or such that $\mi(I)$ is in the interior of the task $i$.

At each recursive call of the algorithm we create a set of artificial tasks, this makes comparing the approximation guarantee to an optimal solution left constrained solution $\overline \OPT_r$ involved and confusing.
So, we will perform a simpler analysis in two steps. 
Namely, first we show that if we are given an optimal left constrained solution $\overline \OPT$ in a subproblem, then its subproblems admit feasible left constrained solutions of the profit that is only by \fab{a} $1+ O(\eps)$ factor smaller that the value of $\overline \OPT$. 
The second part shows that each subproblem indeed recovers an almost optimal solution, i.e., a feasible solution within \fab{a} $1+ O(\eps)$ factor of the optimum.
Applying the two \fab{claims} inductively over the recursive tree proves Lemma \ref{lem:approximation-ratio}.

\begin{lemma}\label{lem:profit1}
Given an optimal left constrained solution $(\overline{\OPT}_I, \{\overline{\OPT}_b\}_{b\in B}, P\fab{(\cdot)})$ in subproblem $I$, there exists a \alex{non-discarded} configuration, with boxes over $I$ denoted by $\{\overline B_{w,d}\}_{w,d}$ and numbers $\{n_b\}_{b \in \overline B_{w,d}}$, numbers $n_b^L$ and $n_b^R$ for $b\in B$, and feasible left constrained solutions $(Q_{I^R}, \{Q^R_b\}_{b\in B^R}, P^R\fab{(\cdot)})$ and $(Q_{I^L}, \{Q^L_b\}_{b\in B^L}, P^L\fab{(\cdot)})$ in the right and left subinstance, respectively, 
such that $w(Q_{I^R}) + \sum_{b \in \overline B_{w,d}} n_b \cdot w(b)  + w(Q_{I^L})\ge (1-O(\eps)) w(\overline \OPT_I) $.  
\end{lemma}
\alexnote{I added the part about not discarding these guesses}
\begin{proof} 
We argue that our guessed boxes and artificial tasks created satisfy the lemma.

First, we show that for a chosen set of boxes $\bar B$ there are sets of tasks in the left subinstance that fit into them. 

Let $\overline{\OPT}^{(L)},\overline{\OPT}^{(M)},$ and $\overline{\OPT}^{(R)}$ as in the algorithm, i.e., the subset of $\overline{\OPT}_I$ scheduled left of $\mi(I)$, containing $\mi(I)$ in the interior, and scheduled right of $\mi(I)$.
By Lemma \ref{lem:boxPacking}, for our choice of boxes $\{\bar B_{w,d}\}_{w,d}$ there exists a subset $S\subseteq \overline{\OPT}^M$ of tasks and a partition of $S$ into $\{S_b\}_{b\in \bigcup_{w,d} \bar B_{w,d}}$ such that the total demand of all boxes does not exceed the demand needed to place the tasks $\overline \OPT ^M$, $w(S) \ge (1 - O(\eps)) w(\overline{\OPT}^M)$, and such that for each $b\in \bigcup_{w,d }\bar B_{w,d}$ the set $S_b$ fits into $b$. 

All of the boxes $\{\bar B_{w,d}\}_{w,d}$ are passed to the left subinstance. Since all the tasks whose time-window contains an edge of $I^L$ are in the left subinstance, so are all the tasks in $S$. 
Naturally, for each $b \in \bigcup_{w,d }\bar B_{w,d}$, we set the corresponding set $Q^L_b$ to $S_b$. This is a valid choice for $n_b \coloneqq |S_b|$.

Finally, as mentioned earlier by Lemma \ref{lem:boxPacking}, total profit from the tasks $\{S_b\}_{b \in \bigcup_{w,d }\bar B_{w,d} }$ is at least $(1-O(\eps))$ fraction of the profit of $\overline{\OPT}^M$.

Next, we show that for a certain choice of numbers $n_b^L$ and $n_b^R$, there are feasible sets $\{Q^L_b\}_{b\in B}$ and $\{Q_b^R\}_{b\in B}$, i.e., for each box $b \in B$ at $I$, there is a choice of  $n_b^L$ and $n_b^R$ summing to $n_b$ such that there are sets $Q^L_b$, $|Q^L_b| = n_b^L$ and $Q^R_b$, $|Q^R_b| = n_b^R$ among the tasks in the left and right subproblem respectively, such that both $Q^L_b$ and $Q^R_b$ fit into $b$.
As we are free to chose $n_b^L$ and $n_b^R$, we can simply set $Q^L_b$ to be the subset of tasks in $\overline \OPT_b$ that are also tasks in the left subproblem, and similarly we take $Q^R_b$ to be the subset of tasks of $\overline \OPT_b$ that are also in the right subproblem. This is valid for $n_b^L \coloneqq |Q^L_b| $ and $n_b^R \coloneqq |Q^R_b| $.
Since every task $i\in T'\setminus \overline \OPT_I$ is either a task in the left or the right subinstance we have $n_b^L+ n_b^R = n_b$.

$Q^L_I$ is set to be exactly the same as $\overline{\OPT}^{L}$, with  $P^L(i) = P (i)$ for each $i \in Q^L_I$. Clearly, the profits of those the sets $Q^L_I$ and 
$\overline{\OPT}^{L}$ are the same. 

It is left to show, how to choose the boxes $\tilde B$, and to show that for a choice of them there are sets $\{Q_b\}_{b\in \tilde B}$ fitting in those boxes, as well as to define  $Q^R_I$ and the respective schedule $P^R$. Here, we will rely on the set $\overline{\OPT}^{R}$ and Lemma \ref{lem:harmonic-grouping}.

We consider two subsets of $\overline{\OPT}^{R}$.
Let $H$ be the subset of tasks in $\overline{\OPT}^{R}$ whose time-window contains $I^R$ strictly (i.e., $H = \overleftarrow{T}' \cap \overline{\OPT}^{R}$). 
\alex{Let $R$ be the other tasks in $\overline{\OPT}^{R}$, i.e. $R=\overline{\OPT}^{R} \setminus H$.} 

The set $Q^R_I$ \alex{is} defined as the union of $R$ and a set of artificial tasks obtained by applying the linear grouping lemma to the tasks $H$. In particular, the tasks in $H$ cannot be a part of $Q^R_I$ as they are passed to the left subproblem. 
Formally, for all tasks of type $w,d$ in $H$, denoted by $H_{w,d}$  let $\tilde H_{w,d}$ be the set of artificial tasks given by applying Lemma \ref{lem:harmonic-grouping} to $H_{w,d}$. 
Then, we set $Q^R_I \coloneqq R \cup \left( \bigcup_{w,d} \tilde H_{w,d}\right)$.
By the third bullet point of the lemma, the tasks $\tilde H_{w,d}$ scheduled by $\tilde P(\cdot)$ can be scheduled instead of $H_{w,d}$ scheduled by $P$ without violating the capacity constraints. 
Formally, we set $P^R(i)$ to $\tilde P(i)$ for each $i\in H_{w,d}$.
By the first bullet point of Lemma \ref{lem:harmonic-grouping} and since all tasks in $H_{w,d}$ have the same profit, we also have $w(\tilde H_{w,d}) \ge (1 - O(\eps)) w(H_{w,d})$. 
Combining with the above, we can already conclude that $w(Q_{I^R}) + \sum_{b \in \overline B_{w,d}} n_b \cdot w(b)  + w(Q_{I^L})\ge (1-O(\eps)) w(\overline \OPT_I)$.

Let $\tilde B$ and $\{n_b\}_{b \in \tilde B}$ the corresponding box and numbers obtained by the above application of Lemma \ref{lem:harmonic-grouping} (these are passed to the left subinstance as boxes). Recall that by construction, all the tasks in $H$ are passed to the left subproblem.
It is left to identify sets that fit into the boxes $\tilde B$. 
Here, we use the fourth bullet point of the lemma.\anote{refer to number instead}
It says that for each $w,d$, there is a collection of tasks $\{Q'_{b}\}_{b \in \tilde B}$ of tasks in $H_{w,d}$ such that $Q'_b$ fits into $b\in \tilde B$ and has size $n_b$.
\alex{Thus we can set $Q^L_b = Q'_b$ for $b \in \tilde B$.
It remains to show that this configuration is not discarded. Consider a task $i$ from the input of the left subproblem. Every task in the input of the left subproblem has either $\tw(i)\subseteq I^L$ or $\tw(i)$ contains the rightmost edge from $I^L$. Thus for $i\in Q^L_b$, $\tw(i)$ contains the rightmost edge from $I^L$. As $\pb(b)\subseteq I^R$ for such a box, $\tw(i)$ also contains edges from $I^R$. This implies that $\tw(i)$ also contains the leftmost edge from $I^R$. Thus the last bullet point of Lemma \ref{lem:harmonic-grouping} applies and yields that this configuration is not discarded, thus completing the proof. }
\end{proof}
\begin{lemma}\label{lem:feasibility}
For a given \alex{non-discarded} configuration, given a
solution for the right and left subinstance $(Q_{I^R}, \{Q^R_b\}_{b\in B^R}, P^R\fab{(\cdot)})$ and $(Q_{I^L}, \{Q^L_b\}_{b\in B^L}, P^{\fab{L}}\fab{(\cdot)})$
the obtained candidate solution $(Q_{I}, \{Q_b\}_{b\in B}, P\fab{(\cdot)})$ is feasible;
moreover $w(Q_I) =  w(Q_{I^R}) + \sum_{b\in \bigcup_{w,d} \bar B_{w,d}}n_b \cdot w(b) + w(Q_{I^L})$.
\end{lemma}
\begin{proof}
Let us recall how the candidate solution is obtained. 
$Q_I$ is defined as a union of: $Q_{I^L}$, the tasks obtained by replacing the artificial tasks in $Q_{I^R}$, all the non-artificial tasks in $Q_{I^R}$ and the tasks in $\{Q^L_b\}_{b\in\bar{B}_{w,d}}$.
For the tasks $i$ in $Q_{I^L}$ and the set of non-artificial tasks in $Q_{I^R}$, we let $P(i) = P^L(i)$ and $P(i) = P^R(i)$ respectively. 

Recall that the boxes present in the left subproblem are $B^L = B\cup \bar B \cup \tilde B$, in the right subproblem are $B^R = B$, and in the current problem are $B$. 

For each $w$ and $d$, each box $b\in\bar{B}_{w,d}$, and each task $i\in Q^L_{b}$ we define a path $P(i)$ for $i$ that a subpath of  $\pb(b)\cap\tw(i)$ starting e.g., at the leftmost edge of that path and having length $p(i)$.
This is feasible since the tasks in $Q^L_{b}$ fit into $b$.
This way, we for each $b\in\bar{B}_{w,d}$ we have added $n_b \cdot w(b)$ profit to the candidate solution at $I$.

Thus, we need to prove why we can use the tasks fitting in the boxes corresponding to the artificial task instead of the artificial tasks, i.e., the tasks $\{Q^L_b\}_{b \in \tilde B}$.
\alex{As the configuration was not discarded, there is a function $f$ that yields the desired replacement map for the artificial tasks: One can schedule $f(i)$ somewhere during $\tw(f(i))\cap P(i)$, as this has length at least $p(f(i))$.}

Finally, for each $b\in B$, we let $Q_b = Q^L_b \cup Q^R_b$. 
The set $Q_b$ trivially fits into $b$ since both $Q^L_b $ and $Q^L_b $ fit into $b$. Additionally, it is also sufficient since $n_b = n^L_b + n^R_b$.

The claim on the profit follows since we have added the sufficient profit for each $b\in\bar{B}_{w,d}$, the set $Q^L_I$ is a subset of $Q_I$, and the artificial tasks are replaced by the same number of real tasks with the same profit.
\end{proof}

\subsection{Proof of Theorem \ref{thr:tw-apx} \ant{and \ref{thr:qptas}}}

\ant{We conclude our analysis by proving Theorem \ref{thr:tw-apx} and \ref{thr:qptas}.}

\begin{proof}[\ant{Proof of Theorem~\ref{thr:tw-apx}}]
    We apply Lemma~\ref{lem:preprocessing} with parameter $\eps' = O(\eps/\log m)$ to our given instance and obtain an instance where $W = \log_{1+\eps'} (n/\eps') = O(\frac{ \log m \log (n \log m)}{\eps})$ and $D = \frac{1}{\eps'}\log_{1+\eps} (n/\eps) = O(\frac{\log m \log(n/\eps)}{\eps})$.
    Then, Theorem~\ref{thr:tw-apx} follows from Lemma~\ref{lem:runningTime} and \ref{lem:approximation-ratio}.
\end{proof}

\ant{
\begin{proof}[Proof of Theorem~\ref{thr:qptas}]
    For instance of \spanUFP{}, we add a set $E'$ of $m$ edges on the right to the initial edge set $E$, and define the capacity on $E'$ as $u(e) = 0$ for every $e\in E'$, and prolongate the time window of every task to be $E \cup E'$.
    Notice that the optimal solutions of the instance on $E'$ and $E\cup E'$ are identical.
    We can then apply the algorithm of Section~\ref{sec:algorithm}.
    By construction, any solution of the latter instance is left constrained, so the algorithm yields a $(1+O(\varepsilon))$-approximate solution to the initial instance with the same running time as for a general \twUFP{} instance.
\end{proof}
}

\section{Hardness of approximation of \spanUFP}
\label{sec:hardness}

In this section, we prove Theorem \ref{thr:apxhard}, that is, show that \spanUFP\ is $\APX$-hard. Our reduction derives from the 3-Dimensional Matching problem, which is discussed in Section \ref{sec:hardness:preliminaries}. We then provide the actual reduction in Section \ref{sec:hardness:reduction}.

\subsection{Preliminaries}\label{sec:hardness:preliminaries}

In the 3-Dimensional Matching problem (3DM) we are given a tripartite hypergraph with node sets $X=\{x_1,\ldots,x_q\}$, $Y=\{y_1,\ldots,y_q\}$, $Z=\{z_1,\ldots,z_q\}$, and a collection of 3-hyperedges $E=\{h_1,\ldots,h_m\}$  where each hyperedge $h_i$ is a triple $(x_i,y_j,z_k)\in X\times Y\times Z$.
Our goal is to compute a hypermatching of maximum cardinality, where a hyper-matching is a subset of hyper-edges $M\subseteq E$ such that there are no two distinct edges $e,f\in M$ that share a common node.
\ant{The special case of 3DM where any node in $X\cup Y \cup Z$ occurs in at least one and at most $k$ hyperedges in $E$ is known as the $k$-bounded 3-Dimensional Matching problem (3DM-$k$).}
Given an instance $K$ of 3DM (or 3DM-$k$), we denote the cardinality of a maximum hyper-matching by $\opt(K)$.
We have the following inapproximability bound.
\fabr{This result also holds when each node is contained in precisely 2 hyper-edges, but I think we do not need it. So maybe we also can use a stronger inapproximability result}

\begin{theorem}[\cite{CC06,kan,papadimitriou1991optimization}]\label{thm:3DM-hardness}
    3DM-$k$ (and 3DM) is $\APX$-hard for $k\geq 3$, and
    it is $\mathsf{NP}$-hard to approximate the solution of 3DM-$2$ to within $\frac{95}{94}$.
\end{theorem}

We will need a variant of a construction and lemma due to \cite{W97}. Given an instance $K$ of 3DM described with the above notation, let $\rho=\max\{29,3q\}$.
We define the following set $Q(K)$ of $3q+|E|$ (possibly negative but non-zero) integers: 
\begin{itemize}\itemsep0pt
\item For each $x_i\in X$, define $x'_i\coloneqq i\rho+1$;
\item For each $y_j\in Y$, define $y'_j\coloneqq j\rho^2+2$;
\item For each $z_k\in Z$, define $z'_k\coloneqq k\rho^3+4$;
\item For each $h_\ell=(x_i,y_j,z_k)\in E$, define $h'_\ell\coloneqq -i\rho-j\rho^2-k\rho^3-7$;
\end{itemize}

\begin{lemma}\label{lem:uniqueSum}
Given any four numbers in $Q(K)$, their sum is exactly $0$ if and only if those numbers are $\{x'_i,y'_j,z'_k,h'_\ell\}$ for some hyper-edge $h_\ell=(x_i,y_j,z_k)\in E$.
\end{lemma}
\begin{proof}
Consider any tuple of $4$ numbers $n_1,n_2,n_3,n_4\in Q(K)$. They can be expressed as $n_i\coloneqq a_i\rho+b_i$, where the $a_i$'s are non-zero (possibly negative) integers and each $b_i$ belongs to $\{1,2,4,-7\}$. Suppose that their sum is $0$.
First notice that $a_1+a_2+a_3+a_4=0$. Indeed, otherwise $|\sum_{i=1}^{4}a_i\rho|\geq \rho$. However $|b_1+b_2+b_3+b_4|\leq 4\cdot 7<\rho$, where we used $\rho\geq 29$. This contradicts the fact that the sum of the considered numbers is $0$. The constraint $a_1+a_2+a_3+a_4=0$ implies $b_1+b_2+b_3+b_4=0$. The only possible way to achieve the latter constraint is to have exactly one number of each type $x'_i$, $y'_j$, $z'_k$, and $h'_\ell$. Assume w.l.o.g.\ $n_1=i\rho+1$, $n_2=j\rho^2+2$, $n_3=k\rho^3+4$, and $n_4=-a\rho-b\rho^2-c\rho^3-7$.
Assume by contradiction $c\neq k$. Then one has that $|k\rho^3-c\rho^3|\geq \rho^3$, while $|i\rho+1+j\rho^2+2+4-a\rho-b\rho^2-7|\leq 2q(\rho+\rho^2)<\rho^3$, where the last inequality holds since $\rho\geq 3q$.
This contradicts $n_1+n_2+n_3+n_4= 0$, hence $c=k$. Assume by contradiction $b\neq j$. Then $|j\rho^2+k\rho^3-b\rho^2-c\rho^3|\geq \rho^2$ and $|i\rho+1+2+4-a\rho-7|\leq 2q\rho<\rho^2$, where the last inequality holds since $\rho\geq 3q$. This contradicts $n_1+n_2+n_3+n_4= 0$, hence $b=j$. Then the only way to have $n_1+n_2+n_3+n_4= 0$ is that $i=a$. The claim follows.
\end{proof}

\subsection{The reduction from 3DM}
\label{sec:hardness:reduction}

The following construction is inspired by \cite{KKW22}, who used a variant of it to show some hardness results for the related \roundUFP\ problem. Consider any instance $K$ of 3DM-$k$ (on $3q$ nodes).
Consider also the corresponding set of numbers $Q(K)$ as in Lemma \ref{lem:uniqueSum}. In particular, for every element $u\in X\cup Y\cup Z\cup E$, we let $u'$ be the associated number. Let $A$ be a large enough integer (depending on $q$) to be fixed later.  We define an instance $\sigma(K)$ of \spanUFP\ as follows:
\begin{enumerate}[label=(\Alph*)]\itemsep0pt
\item\label{item:reduction-intervals} The path $G$ contains $(2A+1)q-1$ edges. These edges are subdivided into $q$ intervals of $2A$ edges each with positive capacity, separated by single edges of capacity $0$.
The edge capacities on each interval are $4A+4$ on the leftmost $A$ edges and $4A$ on the rightmost $A$ edges. 
\item\label{item:reduction-tasks} For each $u\in X\cup Y\cup Z\cup E$, we define two tasks $t_L(u)$ and $t_R(u)$ of weight $1$. Task $t_L(u)$ has length $A-10 u'$ and demand $A+10 u'+1$. Task $t_R(u)$ has length $A+10 u'$ and demand $A-10 u'$. Notice that the $+1$ term appears only in tasks of type $t_L(\cdot)$, which is critical for our arguments.
\end{enumerate}

Let $\mu\coloneqq 1+\max_{u'\in Q(K)} 10 |u'| \leq 10q(\rho+\rho^2+\rho^3)+8$ and $A\coloneqq 5\mu+4$, so that all lengths and demands are positive and polynomially bounded in $q$. Notice that the lengths are demands are between $A-\mu$ and $A+\mu$.

The key property of this reduction is that it essentially preserves the size of an optimal solution, as stated in the following lemma.

\begin{lemma}\label{lem:opt-reduction}
    There is a matching of size $p$ for the instance $K$ of 3DM if and only if there is a solution for $\sigma(K)$ of size $p+7q$.
\end{lemma}

We crucially use that we focus on the bounded case of 3DM to get the following lemma.

\begin{lemma}\label{lem:bound-opt-3dm}
    For any instance $K$ of 3DM-$k$, we have $\opt(K) \geq \frac{1}{3k-2}q$.
\end{lemma}

The proof of Theorem~\ref{thr:apxhard} follows immediately.

\begin{proof}[Proof of Theorem~\ref{thr:apxhard}]
    We provide a PTAS reduction from 3DM-$k$ to \spanUFP. Then, plugging in $k=3$ yields $\APX$-hardness, and plugging in $k=2$ yields the inapproximability bound.

    Let $\calA$ be an $(1-\delta)$-approximation algorithm, that is, for an instance $T$ of \spanUFP{} (in the cardinality case and with polynomially bounded demands and capacities), $\calA(T)$ is a solution with $|\calA(T)| \geq (1-\delta)\opt(T)$, for some fixed $\delta>0$.
    For any instance $K$ of 3DM, let $M_{\calA,K}$ be the solution to $K$ returned from $\calA(\sigma(K))$ as constructed in the proof of Lemma~\ref{lem:opt-reduction}.
    By Lemma~\ref{lem:opt-reduction}, we have
    \begin{equation*}
        |M_{\calA,K}| + 7q \geq |\calA(\sigma(K))| \geq (1-\delta)\opt(\sigma(K)) = (1-\delta)(\opt(K)+7q)
    \end{equation*}
    and therefore, by Lemma~\ref{lem:bound-opt-3dm},
    \begin{align*}
        |M_{\calA,K}|
        &\geq (1-\delta)\opt(K) - 7\delta q \\
        &\geq (1-\delta)\opt(K) - 7\delta(3k-2)\opt(K) \\
        &= (1-(21k-13)\delta)\opt(K).
    \end{align*}
    We therefore obtain a $(1-\varepsilon)$-approximation for 3DM-$k$ for $\varepsilon= \frac{\delta}{21k-13}$ and thus a PTAS reduction.
    By Theorem~\ref{thm:3DM-hardness}, $\calA$ cannot exist for $\delta \leq \frac{1}{95}$ and $k=2$ (unless $\mathsf{P=NP}$), which implies that a $(1-\varepsilon)$-approximation for \spanUFP{} cannot exist for $\varepsilon \leq \frac{1}{2755}$ (unless $\mathsf{P=NP}$).
\end{proof}

We now prove Lemma~\ref{lem:bound-opt-3dm}, then Lemma~\ref{lem:opt-reduction}.

\begin{proof}[Proof of Lemma~\ref{lem:bound-opt-3dm}]
    We construct a hypermatching $M$ greedily: starting with $M=\emptyset$, as long as possible, we add an arbitrary hyperedge $h$ from $E$ to $M$ and remove $h$ and all hyperedges in $E$ that share a node with $h$ from $E$. The size of $M$ is thus number of steps of this procedure, and $M$ is clearly a hypermatching, so $\opt(K) \geq |M|$.
    Since at most $k-1$ hyperedges share a node from $X$, $Y$ or $Z$ with $h$, we remove at most $3(k-1)+1$ edges from $E$ at each step, so the algorithm uses at least $\frac{|E|}{3(k-1)+1}$ steps.
    By the handshaking lemma, since every node occurs in at least one hyperedge, we have $|E| \geq q$ and the claim follows.
\end{proof}

As mentioned before, the construction of $\sigma(K)$ resembles an instance of 2-Dimensional Vector Bin Packing (2-VBP), where the ``bins'' are the intervals of length $2A$ defined in \ref{item:reduction-intervals} where we have to pack as many of the tasks defined in \ref{item:reduction-tasks} as possible. Crucially, in an optimal packing of $\sigma(K)$, $\opt(K)$ many intervals are saturated (i.e., contain 8 tasks each) and all other intervals contain 7 tasks.
This is shown by the two following lemmas.

\begin{lemma}\label{lem:hyperedgeMapping}
Let $T'$ be a subset of tasks in a feasible solution which are scheduled in a given interval $I$. Then $|T'|\leq 8$. Furthermore, $|T'|=8$ if and only if $T'$ consists of $t_L(x_i)$, $t_R(x_i)$, $t_L(y_j)$, $t_R(y_j)$, $t_L(z_k)$, $t_R(z_k)$, $t_L(h_\ell)$, and $t_R(h_\ell)$ for some hyper-edge $h_\ell=(x_i,y_j,z_k)\in E$.
\end{lemma}

\begin{proof}
First we define a feasible schedule of $t_L(x_i)$, $t_R(x_i)$, $t_L(y_j)$, $t_R(y_j)$, $t_L(z_k)$, $t_R(z_k)$, $t_L(h_\ell)$, and $t_R(h_\ell)$, for any hyper-edge $h_\ell=(x_i,y_j,z_k)\in E$ in any interval $I$:
Let the tasks $t_L(x_i)$, $t_L(y_j)$, $t_L(z_k)$ and $t_L(h_\ell)$ start at the leftmost edge of $I$ and let the remaining tasks $t_R(x_i)$, $t_R(y_j)$, $t_R(z_k)$ and $t_R(h_\ell)$ end at the rightmost edge of $I$.
Notice that for $a\in \{x_i, y_j, z_k, h_\ell\}$, $t_L(a)$ and $t_R(a)$ do not overlap and that $t_L(x_i)$, $t_L(y_j)$, $t_L(z_k)$ do not intersect the rightmost $A$ edges of $I$ and $t_R(h_\ell)$ does not intersect the rightmost $A$ edges of $I$.
Therefore, by Lemma~\ref{lem:uniqueSum}, on the $A$ leftmost edges, the demand is a most $(A+10h'_\ell+1)+(A+10x'_i+1)+(A+10y'_j+1)+(A+10z'_k+1)=4A+4$ (since $t_L(h_\ell)$ covers all $A$ rightmost edges), while on the $A$ rightmost edges, the demand is a most $(A-10h'_\ell)+(A-10x'_i)+(A-10y'_j)+(A-10z'_k)=4A$ (since $t_R(x_i)$, $t_R(y_j)$ and $t_R(z_k)$ cover the $A$ rightmost edges entirely).
The schedule is thus feasible.

We now argue that any feasible schedule on $I$ of $8$ or more tasks must be like above.
Let $e_{L'}$ be the $(A-\mu)$-th edge from the left and let $e_{R'}$ be the $(A+2\mu)$-th edge from the left. As every task has length at least $A-\mu$, every scheduled task lies on $e_{L'}$ or $e_{R'}$. Let $L'$ be all tasks lying on $e_{L'}$ and $R'$ be all tasks lying on $e_{R'}$. Thus $T'=L'\cup R'$.

Note that every task has demand of at least $A-\mu$. As the maximal capacity of any edge is at most $4A+4<5(A-\mu)$, at most 4 tasks can be placed on any edge. For the edges $e_{L'}$ and $e_{R'}$ this yields $|L'|\leq 4$ and $|R'|\leq 4$, implying  $|T'|\leq 8$.

From now on consider the case $|T'|=8$. Then $|L'|=|R'|=4$ and $L'\cap R'=\emptyset$. Let $L'=\{t_1,t_2,t_3,t_4\}$ and $R'=\{t_5,t_6,t_7,t_8\}$. Let $d(t_i)\coloneqq A-\delta_i+\gamma_i$ and $\ell(t_i)=A+\delta_i$, where $\delta_i$ is a multiple of $10$ and $\gamma_i\in \{0,1\}$ (in particular $\gamma_i=1$ if $t_i$ is of type $t_L(\cdot)$ and $\gamma_i=0$ otherwise).  Define $\Delta_L=\delta_1+\delta_2+\delta_3+\delta_4$, $\Delta_R=\delta_5+\delta_6+\delta_7+\delta_8$ and $\Delta=\Delta_L+\Delta_R$.
Note that the capacity of $e_{L'}$ is $4A+4$ as $e_{L'}$ is within the $A$ leftmost edges and the capacity of $e_{R'}$ is $4A$ as it is within the $A$ rightmost edges.
Thus $\sum_{i=1}^4 d(t_i)\leq 4A+4$ and $\sum_{i=5}^8 d(t_i)\leq 4A$ or, equivalently, $-\Delta_{L}+\sum_{i=1}^4\gamma_i\leq 4$ and $-\Delta_{R}+\sum_{i=5}^8\gamma_i\leq 0$.
We have $\Delta_L \geq 0$: since $\Delta_L$ is a multiple of 10, $\Delta_L < 0$ would imply $\Delta_L \leq -10$ which contradicts $-\Delta_{L}+\sum_{i=1}^4\gamma_i\leq 4$. $\Delta_R \geq 0$ holds for the same reason.

Assume w.l.o.g.\ that $L'$ is ordered by ending edges (with the one from $t_4$ being the rightmost) and that $R'$ is ordered by starting edges (with the one from $t_8$ being the rightmost). As argued before, at most $4$ tasks can overlap on any given edge.
This holds for the sets of tasks $\{t_1,t_2,t_3,t_4,t_5\}$, $\{t_2,t_3,t_4,t_5,t_6\}$, $\{t_3,t_4,t_5,t_6,t_7\}$ and $\{t_4,t_5,t_6,t_7,t_8\}$, meaning that $t_i$ and $t_{i+4}$ cannot overlap on any edge for every $i\in \{1,2,3,4\}$. As they are both scheduled within the same interval of length $2A$, this implies 
\begin{equation}\label{eq:lengths-delta}
	2A + \delta_i + \delta_{i+4} = \ell(t_i) + \ell(t_{i+4}) \leq 2A \quad \text{for every $i \in \aw{\{1,2,3,4\}}$}
\end{equation}
and thus
$\Delta=\Delta_L+\Delta_R=\delta_1+\ldots+\delta_8\leq 0$.
It follows that $\Delta_L = \Delta_R = 0$ and that the inequalities in (\ref{eq:lengths-delta}) are tight, i.e., $\delta_i = -\delta_{i+4}$ for $i \in \{1,2,3,4\}$.
From $\Delta_R = 0$, we also get $\gamma_5=\gamma_6=\gamma_7=\gamma_8=0$. Therefore, by Lemma~\ref{lem:uniqueSum}, $R'=\{t_R(x_i),t_R(y_j),t_R(z_k),t_R(h_\ell)\}$ for some hyper-edge $h_\ell=(x_i,y_j,z_k)$.
Since $\delta_i = -\delta_{i+4}$ for $i\in \{1,2,3,4\}$, we have $L'=\{t_L(x_i),t_L(y_j),t_L(z_k),t_L(h_\ell)\}$.
\end{proof}

\begin{figure}
	\centering
	\begin{tikzpicture}
		\def\xi{0.2};
		\def\yj{0.5};
		\def\zk{0.8};
		\def\a{5};
		\def\yscale{0.2};
		\def\eps{0.25};
		\def\hl{-1.5};
		\draw [draw=black, fill=orange!10!white] ({\a-\hl}, {\yscale*(\a+\hl+\eps)}) rectangle ({0},{0});
		\node[] at ({(\a-\hl)/2},{\yscale*(\a+\hl+\eps)/2}) {$t_L(h_\ell)$};
		\draw [draw=black, fill=yellow!10!white] ({\a-\xi}, {\yscale*(2*\a+\hl+\xi+2*\eps)}) rectangle ({0},{\yscale*(\a+\hl+\eps)});
		\node[] at ({(\a-\xi)/2},{\yscale*(3*\a+2*\hl+\xi+3*\eps)/2}) {$t_L(x_i)$};
		\draw [draw=black, fill=green!10!white] ({\a-\yj}, {\yscale*(3*\a+\hl+\xi+\yj+3*\eps)}) rectangle ({0},{\yscale*(2*\a+\hl+\xi+2*\eps)});
		\node[] at ({(\a-\yj)/2},{\yscale*(5*\a+2*\hl+2*\xi+\yj+5*\eps)/2}) {$t_L(y_j)$};
		\draw [draw=black, fill=blue!10!white] ({\a-\zk}, {\yscale*(4*\a+\hl+\xi+\yj+\zk+4*\eps)}) rectangle ({0},{\yscale*(3*\a+\hl+\xi+\yj+3*\eps)});
		\node[] at ({(\a-\zk)/2},{\yscale*(7*\a+2*\hl+2*\xi+2*\yj+\zk+7*\eps)/2}) {$t_L(z_k)$};
		
		\draw [draw=black, fill=orange!15!white] ({\a-\hl}, {0}) rectangle ({2*\a},{\yscale*(\a-\hl)});
		\node[] at ({(3*\a-\hl)/2},{\yscale*(\a-\hl)/2}) {$t_R(h_\ell)$};
		\draw [draw=black, fill=yellow!15!white] ({\a-\xi}, {\yscale*(\a-\hl)}) rectangle ({2*\a},{\yscale*(2*\a-\hl-\xi)});
		\node[] at ({(3*\a-\xi)/2},{\yscale*(3*\a-2*\hl-\xi)/2}) {$t_R(x_i)$};
		\draw [draw=black, , fill=green!15!white] ({\a-\yj}, {\yscale*(2*\a-\hl-\xi)}) rectangle ({2*\a},{\yscale*(3*\a-\hl-\xi-\yj)});
		\node[] at ({(3*\a-\yj)/2},{\yscale*(5*\a-2*\hl-2*\xi-\yj)/2}) {$t_R(y_j)$};
		\draw [draw=black, fill=blue!15!white] ({\a-\zk}, {\yscale*(3*\a-\hl-\xi-\yj)}) rectangle ({2*\a},{\yscale*(4*\a-\hl-\xi-\yj-\zk)});
		\node[] at ({(3*\a-\zk)/2},{\yscale*(7*\a-2*\hl-2*\xi-2*\yj-\zk)/2}) {$t_R(z_k)$};

		\draw[very thick, color=red] (-0.5,0)--(0,0);
		\draw[very thick, color=red] (0,0)--(0,{\yscale*(4*\a+4*\eps)});
		\draw[very thick, color=red] (0,{\yscale*(4*\a+4*\eps)})--(\a,{\yscale*(4*\a+4*\eps)});
		\draw[very thick, color=red] (\a,{\yscale*(4*\a+4*\eps)})--(\a,{\yscale*(4*\a+0*\eps)});
		\draw[very thick, color=red] (\a,{\yscale*(4*\a+0*\eps)})--(3/2*\a,{\yscale*(4*\a+0*\eps)});
		\draw[very thick, color=red] (3/2*\a,{\yscale*(4*\a+0*\eps)})--(3/2*\a,{\yscale*(4*\a)});	
		\draw[very thick, color=red] (3/2*\a,{\yscale*(4*\a)})--(2*\a,{\yscale*(4*\a)});
		\draw[very thick, color=red] (2*\a,{\yscale*(4*\a)})--(2*\a,0);
		\draw[very thick, color=red] (2*\a,0)--(2*\a+0.5,0);
		
	\end{tikzpicture}
\caption{The capacity profile (red) of one interval and the 8 tasks corresponding to one hyperedge $h_\ell=(x_i, y_j, z_k)$ scheduled in this interval.} 
\label{fig:hardnessSchedule}
\end{figure}
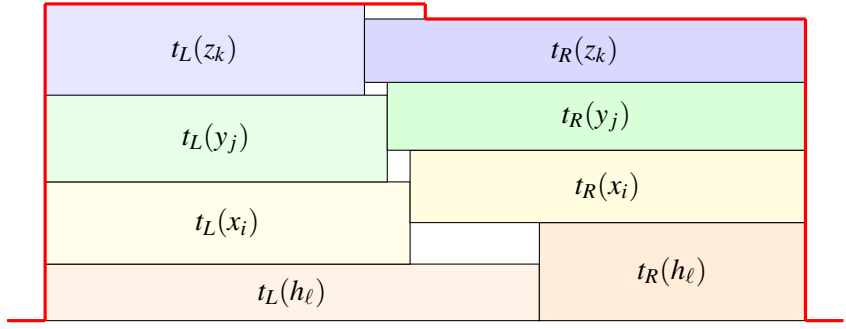

\begin{lemma}\label{lem:almost-full-interval}
    Let $x_i \in X$, $y_j \in Y$, $z_k \in Z$ and $h_\ell \in E$. Then, we can schedule $t_L(x_i)$, $t_R(x_i)$, $t_L(y_j)$, $t_R(y_j)$, $t_L(z_k)$, $t_R(z_k)$ and $t_L(h_\ell)$ or $t_L(x_i)$, $t_R(x_i)$, $t_L(y_j)$, $t_R(y_j)$, $t_L(z_k)$, $t_R(z_k)$ and $t_R(h_\ell)$ in any given interval $I$.
\end{lemma}

\begin{proof}
    The claim is clear if $h_\ell = (x_i,y_j,z_k)$ by Lemma~\ref{lem:hyperedgeMapping}, so assume $h_\ell \neq (x_i,y_j,z_k)$, or equivalently, $\tau \coloneqq x_i'+y'_j+z'_k+h'_\ell \neq 0$ by Lemma~\ref{lem:uniqueSum}.

    Assume first $\tau < 0$.
    We schedule $t_L(x_i)$, $t_L(y_j)$, $t_L(z_k)$ and $t_L(h_\ell)$ on the leftmost edge of $I$ and schedule $t_R(x_i)$, $t_R(y_j)$ and $t_R(z_k)$ on the rightmost edge of $I$. Notice that for any $a \in \{x_i, y_j, z_k\}$, $t_L(a)$ does not overlap with $t_R(a)$.
    Clearly, on any of the $A$ leftmost edges of $I$, the total demand is upper bounded by $t_L(x_i) + t_L(y_j) + t_L(z_k) +t_L(h_\ell) = 4A + 4 + 10\tau < 4A+4$.
    On any of the $A$ rightmost edges of $I$, the total demand is upper bounded by $4A$ since the demands of $t_R(x_i)$, $t_R(y_j)$, $t_R(z_k)$ and $t_L(h_\ell)$ are upper bounded by $A$ (note that $h_\ell'<-1$) and the other tasks are not scheduled on those edges.
    Therefore, the schedule is feasible.

    Assume now $\tau > 0$.
    We schedule $t_L(x_i)$, $t_L(y_j)$ and $t_L(z_k)$ on the leftmost edge of $I$ and schedule $t_R(x_i)$, $t_R(y_j)$, $t_R(z_k)$ and $t_R(h_\ell)$ on the rightmost edge of $I$. Notice as before that for any $a \in \{x_i, y_j, z_k\}$, $t_L(a)$ does not overlap with $t_R(a)$.
    On any of the $A$ leftmost edges of $I$, the total demand is upper bounded by $t_L(x_i) + t_L(y_j) + t_L(z_k) \leq 3A + 3\mu \leq 4A+4$.
    On any of the $A$ right edges of $I$, the total demand is upper bounded by $t_R(x_i) + t_R(y_j) + t_R(z_k) + t_R(h_\ell) = 4A - 10\tau < 4A$.
    Therefore, the schedule is feasible.
\end{proof}

We can finally show Lemma~\ref{lem:opt-reduction}.

\begin{proof}[Proof of Lemma~\ref{lem:opt-reduction}]
    Let $M = \{h^{(1)}, \dots, h^{(p)}\}$ be a feasible hypermatching of size $p$ for an instance $K$ of 3DM.
    We construct a schedule $S \subseteq \sigma(K)$ with $|S| = p + 7q$.
    For every hyperedge $h^{(k)} = (x^{(k)},y^{(k)},z^{(k)}) \in M$, schedule $t_L(x^{(k)})$, $t_R(x^{(k)})$, $t_L(y^{(k)})$, $t_R(y^{(k)})$, $t_L(z^{(k)})$, $t_R(z^{(k)})$, $t_L(h^{(k)})$, and $t_R(h^{(k)})$ in the $k$-th interval, which is doable by Lemma~\ref{lem:hyperedgeMapping}. From the leftover nodes and hyperedges, form $p-q$ many quadruples $(x^{[k]},y^{[k]},z^{[k]},h^{[k]}) \in X\times Y \times Z \times E$ arbitrarily (for $k \in \{p+1, \dots, q\}$). Then, for every $k$, schedule either $t_L(x^{[k]})$, $t_R(x^{[k]})$, $t_L(y^{[k]})$, $t_R(y^{[k]})$, $t_L(z^{[k]})$, $t_R(z^{[k]})$ and $t_L(h^{[k]})$ or $t_L(x^{[k]})$, $t_R(x^{[k]})$, $t_L(y^{[k]})$, $t_R(y^{[k]})$, $t_L(z^{[k]})$, $t_R(z^{[k]})$, and $t_R(h^{[k]})$ in the $k$-th interval, which is doable by Lemma~\ref{lem:almost-full-interval}.
    The resulting schedule $S$ is clearly feasible and contains $8p + 7(q-p) = p + 7q$ tasks.

    To prove the converse, notice first that $\sigma(K)$ always has a solution of size $7q$ by scheduling $7$ tasks in every interval like in Lemma~\ref{lem:almost-full-interval}.
    Now, if there is a feasible schedule $S \subseteq \sigma(K)$ with $p + 7q$ tasks (and $p\geq 0$), by the pigeonhole principle and because at most $8$ tasks fit inside any interval by Lemma~\ref{lem:hyperedgeMapping}, at least $p$ intervals must contain $8$ tasks.
    Let $\calI$ be a set of all intervals in which $8$ tasks are scheduled, respectively, and let $M_{S,K}$ be the set of hyperedges $h_\ell \in E$ such that $t_L(h_\ell)$ or $t_R(h_\ell)$ is scheduled in an interval in $\calI$.
    By Lemma~\ref{lem:hyperedgeMapping}, since all intervals in $\calI$ contain $8$ tasks, each such interval must contain both $t_L(h_\ell)$ and $t_R(h_\ell)$ for some $h_\ell \in E$ (which is also in $M_{S,K}$ by definition), but not $t_L(h_{\ell'})$ or $t_R(h_{\ell'})$ for any $\ell' \neq \ell$. Therefore, $|M_{S,K}| = |\calI| \geq p$.
    Now take $h_\ell = (x_i, y_j, z_k)\in M_{S,K}$ and $h_{\ell'} = (x_{i'}, y_{j'}, z_{k'}) \in M_{S,K}$, with $\ell \neq \ell'$.
    By Lemma~\ref{lem:hyperedgeMapping}, $t_L(x_i)$, $t_R(x_i)$, $t_L(y_j)$, $t_R(y_j)$, $t_L(z_k)$ and $t_R(z_k)$ are scheduled in the same interval as $t_L(h_\ell)$, and $t_R(h_\ell)$, while $t_L(x_{i'})$, $t_R(x_{i'})$, $t_L(y_{j'})$, $t_R(y_{j'})$, $t_L(z_{k'})$, $t_R(z_{k'})$ are scheduled in the same interval as $t_L(h_{\ell'})$ and $t_R(h_{\ell'})$.
    Therefore, $i\neq i'$, $j\neq j'$ and $k\neq k'$, so $M_{S,K}$ is a feasible hypermatching of cardinality (at least) $p$.
\end{proof}

%
%
%
\bibliographystyle{plain}
\bibliography{ipco2025_references}

\newpage

\appendix

\section{Proof of Lemma~\ref{lem:preprocessing}}\label{sec:proof-preprocessing}

\reductions*

\ant{We describe the preprocessing part of our algorithm and thereby prove Lemma~\ref{lem:preprocessing}.}

\begin{lemma}\label{lem:reductionW}
    Let $\alpha \ge 1$. Assume that there is an $\alpha$-approximation algorithm for \twUFP{} in which $w(i)\in [1, n/\eps\fab{]}$ and $w(i)$ is a power of $1+\eps$ running in time $T(n,\eps)$.
   Then there is an $(1+3\eps)\alpha$-approximation algorithm for \twUFP{} running in time $\poly (n) \cdot T(n,\eps)$.
\end{lemma}
\begin{proof}
\fab{Assume w.l.o.g. that $\eps\leq 1/2$.} Let $w^*$ be  the largest profit  in some optimum solution $(\OPT,P^*(\cdot))$ (there are at most $n$ possible values for $w^*$),
and discard all the tasks with profit larger than $w^*$ and with profit smaller than $\eps w^*/n$. This way, we loose at most a fraction $\eps$ of the optimal profit. 
By scaling, we assume that the smallest profit is $1$ and that the largest profit is at most $n/\eps$. Then we round down each profit to the next value of type $(1+\eps)^q$, $q\in \mathbb{N}$. The rounding down reduces the profit by a factor $(1+\eps)$ at most. 
\fab{The overall approximation factor is $\frac{1+\eps}{1-\eps}\alpha\leq (1+3\eps)\alpha$}.  
\end{proof}

Let $U$ be the largest edge capacity. The next lemma shows that we can assume w.l.o.g. that $U$ is polynomially bounded in the number $n$ of tasks.

\begin{lemma}\label{lem:reductionU}
Let $\alpha \ge 1$.
Suppose that there is an $\alpha$-approximation algorithm with $(1+\eps)$ resource augmentation for \twUFP\ running in time $T(n,\eps)$ under the assumption that $U\leq (\frac{n}{\eps})^{1/\eps}$. Then there is an $\frac{\alpha}{1-\eps}$-approximation algorithm for \twUFP\ (with no restriction) with $(1+3\eps)$ resource augmentation running in time $\poly(n)+O(T(n,\eps)\frac{\log U}{\log n})$.\fabr{Removed the $1/\eps$ factor. I think it is rather $\eps$, so we can neglect it (as done in the proof actually)}  
\end{lemma} 
\begin{proof}
Let $r\in \{0,\ldots,1/\eps-1\}$ to be fixed later. We define the following ranges of demands: 
$[A_0,B_0)=[1,(\frac{n}{\eps})^r)$, 
$[A_1,B_1)=[(\frac{n}{\eps})^{r+1},(\frac{n}{\eps})^{r+1/\eps})$, 
$[A_2,B_2)=[(\frac{n}{\eps})^{r+1+1/\eps},(\frac{n}{\eps})^{r+2/\eps})$, ..., 
$[A_q,B_q)=[(\frac{n}{\eps})^{r+1+(q-1)/\eps},(\frac{n}{\eps})^{r+q/\eps})$, 
where $q=O(\log_{(n/\eps)^{1/\eps}} U)=O(\frac{\log U}{\log n})$ is the smallest integer such that $(\frac{n}{\eps})^{r+q/\eps}\geq U$. 
Let $G_j$ be the subset of tasks (\emph{group}) with demand in the range $[A_j,B_j)$. We delete all the tasks not contained in any group $G_j$. Clearly there is a choice of $r$ such that this deletion reduces the profit of some fixed optimal solution $\OPT$ at most by an $\eps$ fraction. We guess this value of $r$ by trying all the possible $1/\eps$ options, and assume it in the following. 

We define a \twUFP\ instance for each such group $G_j$ separately as follows. We set to $0$ all the edge capacities smaller than $A_j$: notice that the corresponding edges cannot accommodate any task in $G_j$ anyway. Furthermore we scale down to $nB_j$ any demand larger than that value: observe that the tasks in $G_j$  cannot use capacity larger than $nB_j$. We next scale demands and capacities so that the smallest non-zero demand is $1$. Observe that in the resulting instance the largest capacity $U$ is at most $n\frac{B_j}{A_j}\leq (\frac{n}{\eps})^{1/\eps\fab{-1}}$. To this instance we apply the algorithm for \twUFP\ as in the assumption, hence getting a solution $\APX_j$. We return the union $\APX$ of the solutions obtained this way.

Clearly the running time of the overall procedure matches the claim.  Consider next the approximation factor. Let $\OPT_j:=\OPT\cap G_j$. By the choice of $r$ one has $\sum_{j}w(\OPT_j)\geq (1-\eps)w(\OPT)$. 
Notice that $\OPT_j$ is a feasible solution (without resource augmentation) for the subproblem corresponding to $G_j$. Thus $w(\APX_j)\geq \frac{1}{\alpha} w(\OPT_j)$. 
Altogether $w(\APX)=\sum_{j}w(\APX_j)\geq \frac{1-\eps}{\alpha}w(\OPT)$. 

It remains to upper bound the needed resource augmentation. Consider any edge $e$. Observe that $\fab{u(e)}\in [\frac{A_j}{1+\eps},\frac{n}{\eps}\frac{B_j}{1+\eps})$ for some group $G_j$. Notice that no task in $\APX$ using edge $e$ can belong to some group $G_{j'}$ with $j'>j$, since this task would have demand at least $\frac{n}{\eps}B_j$ (hence it would not fit even exploiting resource augmentation). By construction the tasks in $\APX_j:=\APX\cap G_j$ satisfy $d(\APX_j\cap T_e)\leq (1+\eps)\fab{u(e)}$. Each one of the remaining tasks $\APX'$ in $\APX$ that use edge $e$ has demand at most $\frac{\eps}{n}A_j$. Hence $d(APX'\cap T_e)\leq \eps A_j\leq \eps(1+\eps)\fab{u(e)}\leq 2\eps \fab{u(e)}$. Altogether, $d(\APX\cap T_e)\leq (1+3\eps)\fab{u(e)}$.
\end{proof}

\begin{proof}[Proof of Lemma \ref{lem:preprocessing}]
\fab{Assume w.l.o.g.\ $\eps\leq 1/5$.} Let\fabr{Added this proof to be more formal} us chain the reductions from Lemmas \ref{lem:reductionW} and \ref{lem:reductionU}, and then round up each demand to the next power of $(1+\eps)$ in each subproblem. We solve each such subproblem with the given algorithm and return the union of the obtained solutions. The approximation factor of the overall algorithm is $\frac{1+3\eps}{1-\eps}\alpha\leq (1+5\eps)\alpha$ and the amount of used resource augmentation is $(1+\eps)(1+3\eps)\leq 1+4\eps$. The running time of the overall process is $\poly(n) T(n,\eps)O\left(\frac{\log U}{\log n}\right)$.
\end{proof}

\end{document}